\documentclass[onecolumn]{IEEEtran}          % twocolumn
\usepackage{graphicx}
\makeatletter
\def\ps@headings{%
\def\@oddhead{\mbox{}\scriptsize\rightmark \hfil \thepage}%
\def\@evenhead{\scriptsize\thepage \hfil \leftmark\mbox{}}%
\def\@oddfoot{}%
\def\@evenfoot{}}
\makeatother
\usepackage{cite}
\usepackage{todonotes}
\usepackage{setspace}
\usepackage{dsfont}
\usepackage{enumerate}
\usepackage{booktabs}
\usepackage{bm}
\usepackage{url}
\usepackage{algorithmicx, algorithm, algpseudocode}
\usepackage{verbatim}
\usepackage{graphicx}
\usepackage{subfigure}
\usepackage{amsmath,amsfonts,amsthm}
\usepackage{scrextend}
\usepackage[misc,geometry]{ifsym} 
\newtheorem{theorem}{\textbf{Theorem}}
\newtheorem{lemma}{\textbf{Lemma}}
\newtheorem{definition}{\textbf{Definition}}
\newtheorem{corollary}{\textbf{Corollary}}
\newtheorem{proposition}{\textbf{Proposition}}
\newtheorem{IP}{Integer Program}
\usepackage{color}
\usepackage[round]{natbib}
\usepackage[T1]{fontenc}
\begin{document}
\title{Vulnerability of clustering under node failure in complex networks%\thanks{Grants or other notes
}


% \author{Alan Kuhnle         \and
%         Nam P. Nguyen       \and
%         Thang N. Dinh       \and
%         My T. Thai
% }
\author{Alan Kuhnle, Nam P. Nguyen, Thang N. Dinh, My T. Thai \\
  Department of Computer and Information Science and Engineering, University of Florida, USA \\
  kuhnle@ufl.edu
}


% \institute{Alan Kuhnle \at
%   Department of Computer and Information Science and Engineering, University of Florida, USA \\
%               \email{kuhnle@ufl.edu}           
%            \and
%            Nam P. Nguyen \at
%            Department of Computer and Information Sciences, Towson University, USA. \\
%            \email{npnguyen@towson.edu}
%            \and
%            Thang N. Dinh \at
%            Department of Computer Science, Virginia Commonwealth University, USA. \\
%            \email{tndinh@vcu.edu}
%            \and
%            My T. Thai \Letter \at
% %           Division of Algorithms and Technologies for Networks Analysis, Faculty of Information Technology, Ton Duc Thang University, Ho Chi Minh City, Vietnam\\
%            Department of Computer and Information Science and Engineering, University of Florida, USA \\
%            \email{%thaitramy@tdt.edu.vn; 
%            mythai@cise.ufl.edu}
% }


\maketitle

\begin{abstract} 
Robustness in response to unexpected events is always desirable
for real-world networks.
To improve the robustness of any networked system, it is important to 
analyze vulnerability to external perturbation such as 
random failures or adversarial attacks occurring to elements of the network.
In this paper, we study an emerging problem in assessing the robustness of 
complex networks: the vulnerability of the clustering of the network to
the failure of network elements.
Specifically, we identify vertices whose failures will 
critically damage the network by degrading its clustering, 
evaluated through the average clustering coefficient.
This problem is important because any significant change made to the clustering, 
resulting from element-wise failures, could degrade network performance
such as the ability for information to propagate in a social network.
We formulate this vulnerability analysis as an optimization problem, 
prove its NP-completeness and non-monotonicity, and
we offer two algorithms to identify the vertices most important to clustering.
Finally, we conduct comprehensive experiments in synthesized 
social networks generated by various well-known models as well 
as traces of real social networks.
The empirical results over other competitive strategies show the 
efficacy of our proposed algorithms.
%\keywords{Network Clustering \and Robustness \and Network vulnerability \and Social networks}

\end{abstract}

\section{Introduction} \label{ssintroduction}

Network resilience to attacks and failures has been a growing concern in recent times.
Robustness is perhaps one of the most desirable properties for corporeal complex networks, 
such as the World Wide Web, transportation networks, 
communication networks, biological networks and social information networks. 
Roughly speaking, robustness of a network evaluates how much the network's normal 
function is affected in case of external perturbation, i.e., it measures the resilience 
of the network in response to unexpected events such as 
adversarial attacks and random failures \citep{holme2002attack}.
Complex systems that can sustain their organizational structure, 
functionality and responsiveness under such unexpected perturbation 
are considered more robust than those that fail to do so.
The concept of \textit{vulnerability} has generally been used to realize and characterize the lack of robustness and resilience of complex systems \citep{criado2012Strutural}.
In order to improve the robustness of real-world systems, it is therefore important to obtain key insights into the structural vulnerabilities of the networks representing them.
A major aspect of this is to analyze and understand the effect of failure (either intentionally or at random) of individual components on the degree of clustering in the network.

Clustering is a fundamental network property
that has been shown to be relevant 
to a variety of topics. For example, consider
the propagation of information through a social 
network, such as the spread of a rumor. 
A growing body of 
work has identified the importance of clustering to such propagation; 
the more clustered a network is, the easier it is for information to 
propagate \citep{centola2010spread, barclay2013peer,lu2011small,malik2013role,centola2011experimental}. In addition, in Fig. \ref{fig:rel_gcc}, we show experimentally a 
strong relationship between the final spread of information and 
the level of clustering in the network, with higher clustering corresponding
to higher levels of expected spread. 
The importance of clustering is not limited to social networks;
in the context of air transportation networks,
\citet{Ponton2013}
argued that higher clustering of such a network is beneficial, as
passengers for a cancelled flight can be rerouted more easily.
In this work, we use average clustering coefficient (ALCC) as
our definition and measure of clustering in a network.
ALCC was proposed for this purpose by \citet{watts1998cds}.

%% As empirical examples of this relationship,
%% the spread of a health behavior in an online social
%% network was shown to increase with clustering \citep{centola2010spread}, 
%% and similarly for the spread of 
%% opinions in
%% a social environment \citep{malik2013role}.  

%%% if a statement like this is kept, need to add graphs demonstrating: 
%Furthermore, in this work we present evidence for a relationship between ALCC and size of cascading failure.
\begin{figure}
  \centering
  \includegraphics[scale=0.36]{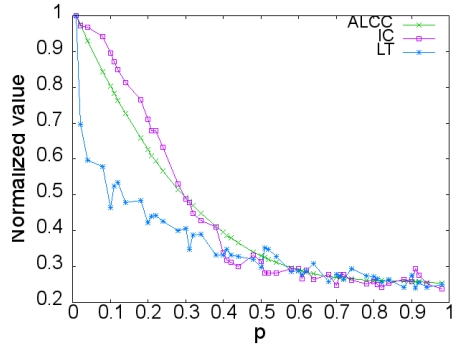}
  \caption{{ \color{black} Relationship between the value of ALCC and the expected number of activations
    under the LT and IC models, normalized by initial value. For more details and 
    discussion
    of $p$, see Section \ref{sect:rel_gcc}. }} \label{fig:rel_gcc}
\end{figure}

The identification of elements that crucially affect the clustering of the network, 
as a result, is of great impact. For example, as a matter of homeland security,
the critical elements for clustering in homeland communication networks
should receive greater resources for protection; in complement, 
the identification of critical elements in a social network of adversaries
could potentially limit the spread of information in such a network.
However, most studies of network vulnerability in the literature focus on how the network behaves when its elements (nodes and edges) are removed based on the pair-wise connectivity \citep{thangton2012}, natural connectivity \citep{hau2014}, or using centrality measures, such as degrees, betweeness \citep{Albert00theinternet}, the geodesic length \citep{holme2002attack}, eigenvector \citep{allesina2009}, etc.
To our knowledge, none of the existing work has examined the average clustering coefficient from the perspective of
vulnerability - as evidenced by the examples above, the damage made to the average clustering, resulted from element-wise failures, can potentially have severe effects on the functionality of the network. %change the entire organization and lead to the unpredictable malfunction or corruption of the whole network.
This drives the need for an analysis of clustering vulnerabilities in complex networks.

Finding a solution for this emerging problem, nevertheless, is fundamentally yet technically challenging because (1) the behavior of ALCC is not monotonic with respect to
node removal and thus can be unpredictable even in response to minor changes, and (2) given large sizes of real networks, the $NP$-completeness of the problem 
prohibits the tractable computation of an exact solution.
In this paper, we tackle the problem and analyze the vulnerabilities of the network clustering. Particularly, we ask the question: 

%----------------------%
\begin{addmargin}[1em]{1em}
\textit{``Given a complex network and its clustering coefficient, what are the most important vertices whose failure under attack, either intentionally or at random, will maximally degrade the network clustering?''}
\end{addmargin}
%----------------------%
%% To measure the network clustering, we utilize the average clustering coefficient (ALCC), a measure proposed in complex network science, as the main objective.
There are many advantages of ALCC over other structural measures \citep{watts1998cds}: (1) it is one of the most popular metrics for evaluating network clustering - the higher the ALCC of a network the better clustering  it exhibits, (2) it implies multiple network modular properties such as small-world scale-free phenomena, small diameter and modular structure (or community structure), and (3) it is meaningful on both connected and disconnected as well as dense and sparse graphs: Sparse networks are expected to have small clustering coefficient whereas extant
complex networks are found to have high clustering coefficients.

Our contributions in this paper are: (1) We define the Clustering  Vulnerability Assessment (CVA) on complex networks, and formulate it as an optimization problem with ALCC as the objective function.
(2) We study CVA's complexity (NP-completeness), provide rigorous proofs and vulnerability analysis on random failures and targeted attacks.
To our knowledge, this is the first time the problem and the analysis are studied specifically for ALCC.
(3) Given the intractability of the problem, we provide two efficient algorithms which scale to large networks to identify the worst-case scenarios of adversary attacks.
Finally, (4) we conduct comprehensive experiments in both synthesized networks (generated by various well-known models) as well as real networks.
The empirical results over other methods show the efficacy and scalability of our proposed algorithms.

The paper is organized as follows: Section \ref{ssrelatedwork} reviews studies that are related to our work. Section \ref{sspreliminaries} describes the notations, measure functions and the problem definition. Section \ref{ssanalysisComplexity} shows the proof of NP-completeness implying the intractability of the problem. Section \ref{ssRandomFailures} and \ref{ssnodeRemoval} present our analysis of clustering behaviors on random failures and targeted attacks, respectively. 
{\color{black} In Section \ref{sect:rel_gcc}, 
we provide further evidence for a correlation between the extent of influence propagation 
and ALCC.} In Section \ref{ssexperiment}, we report empirical results of our approaches in comparison with other strategies. Finally, Section \ref{ssconclusion} concludes the paper.

\section{Related work} \label{ssrelatedwork}
%\vspace{-0.1in}
Vulnerability assessment has attracted a large amount of attention from the network science community. 
Work in the literature can be divided into two categories: Measuring the robustness and manipulating the robustness of a network. 
In measuring the robustness, different measures and metrics have been proposed such as the 
graph connectivity \citep{thangton2012}, the diameter, relative size of largest components, and average size of the 
isolated cluster \citep{Albert00theinternet}. Other work suggests using the 
minimum node/edge cut \citep{Frank1970} or the second smallest non-zero eigenvalue or the Laplacian matrix \citep{Fiedler73}. 
In terms of manipulating the robustness, different strategies has been proposed such as \citet{Albert00theinternet,peixoto2012}, or using graph percolation \citep{Callaway2000}. 
Other studies focus on excluding nodes by centrality measures, such as betweeness and the geodesic length \citep{holme2002attack}, 
eigenvector \citep{allesina2009}, 
the shortest path between node pairs \citep{grubesic08}, or the total pair-wise connectivity \citep{thangton2012}. {\color{black}
\cite{Veremyev2015,Veremyev2014} developed integer programming frameworks to determine
the critical nodes that minimize a connectivity metric subject to a budgetary constraint.
For more information on network vulnerability assessments,
the reader is referred to the surveys \citep{Chen2016} and \citep{Gomes2016} and references therein.}

The vulnerability of the average clustering of a complex network 
has been a relatively unexplored area. 
In a related work \citep{namASONAM13}, the authors introduced the community structure vulnerability to analyze how the communities are affected when top $k$ vertices are excluded from the underlying graphs. They further provided different heuristic approaches to find out those critical components in modularity-based community structure. \citet{namWI14} suggested a method based on the generating edges of a community to find out the critical components. In a similar vein, \citet{alimASONAM14} studied the problem of breaking all density-based communities in the network, proved its NP-hardness and suggested an approximation as well as heuristic solutions. These studies, 
while forming the basis of community-based vulnerability analysis, face a 
fundamental limitation due to the ambiguity of definitions of a community in a network. 
Our work overcomes this particular shortcoming as ALCC is a well-defined and commonly accepted concept for quantifying the clustering of a network. {\color{black} \cite{Ertem2016} studied the problem of how to
detect groups of nodes in a social network with high clustering coefficient; however, their
work does not consider the vulnerability of the average clustering coefficient of a network.}
The diffusion of information in a social network has been studied from many
perspectives, including worm containment \citep{Nguyen2010}, viral marketing \citep{Kempe2003,Dinh2012c,Dinh2013,Kuhnle2017}, and the detection of overlapping communities \citep{Nguyen2011}.

\section{Notations and Problem definition} \label{sspreliminaries}
%%\vspace{-0.1in}
%----------------------%
\subsection{Notations}
%%\vspace{-0.1in}
%----------------------%
Let $G=(V,E)$ be an undirected graph representing a complex network where $V$ is the set of $N$ nodes and $E$ is the set of edges containing $M$ connections.
For a node $u \in V$, denote by $d_u$ and $N(u)$ the degree of $u$ and the set of $u$'s neighbors, respectively. For a subset of nodes $S \subseteq V$, let $G[S]$ and $m_S$ in this order denote the subgraph induced by $S$ in $G$ and the number of edges in this subgraph. Hereafter, the terms ``vertices'' and ``nodes'' as well as ``edges'' and ``links'' are used interchangeably.

\textit{(Triangle-free graphs)} A graph $G$ is said to be \textit{triangle-free} 
if no three vertices of $G$ form a triangle of edges.
Verifying whether a given graph $G$ is triangle-free or not is tractable by computing the 
trace of $A^3$ where $A$ is the adjacency matrix of $G$.
The trace is zero if and only if the graph is triangle-free. This verification 
can be done in polynomial time 
$O(  N^{\omega})$ for $\omega \leq 2.372$ with the latest matrix multiplying result \citep{Gall14}.
{\color{black} Alternatively, one can use the method of \citep{Schank05} with time complexity $O( M^{3/2} )$ to check if the graph is triangle-free.}
%----------------------%
%%\vspace{-0.12in}
\subsection{Clustering Measure Functions}
%%\vspace{-0.1in}
%----------------------%
\subsubsection{Local Clustering Coefficient (LCC)}
%----------------------%
Given a node $u \in V$, there are $d_u$ adjacent vertices of $u$ in $G$ and there are $d_u(d_u-1)/2$ possible edges among all $u$'s neighbors. The local clustering coefficient $C(u)$ is the probability that two random neighbors of $u$ are connected. Equivalently, it quantifies how close 
the induced subgraph of neighbors is to a clique. 
{\color{black}
The local clustering coefficient $C(u)$ is defined \citep{watts1998cds} 
\[ C(u) = \begin{cases} \displaystyle \frac{2 T(u) }{d_u(d_u - 1)} & d_u > 1 \\
                 0 & \text{ otherwise } 
                 \end{cases} \] }
%%%----------------------%
%%\begin{equation}\nonumber
	%%C(u) = \left\{ \begin{array}{ll}
	%%\frac{2 T(u) }{d_u(d_u - 1)} & \text{ if } d_u > 1,\\
	%%0 & \text{ otherwise },
	%%\end{array}\right.
	%%\label{equ:def_lcc}
%%\end{equation}
%%----------------------%
where $T(u)$ is the number of triangles containing $u$. It is clear that $0 \leq C(u) \leq 1$ for any $u \in V$. For any node $v \neq u$, let $\tilde C_v(u)$ denote the clustering coefficient of $u$ in $G[V \backslash \{v\}]$. {\color{black} Finally, define $tr(u,v)$ as the number of triangles containing both vertices $u$ and $v$.}
%Alternatively, $C(u) = \displaystyle \frac{2m_{N(u)}}{d_u(d_u-1)} \text{ when $d_u > 1$}$
%----------------------%
%%\begin{equation}
	%%C(u) = \frac{2m_{N(u)}}{d_u(d_u-1)} \text{ when $d_u > 1$,}
	%%\label{eq:def_lcc2}
%%\end{equation}
%----------------------%
%because $T(u) = m_{N(u)}$.
%----------------------%
\begin{table}[t]%\scriptsize
%\vspace{-0.1in}
  \centering
  \caption{List of Symbols}
	%\vspace{-0.15in}
    \begin{tabular}{ll}
    \addlinespace
    \toprule
    Notation  &  Meaning \\
    \midrule
		%$V$ & Set of vertices/nodes in the graph\\
		%$E$ & Set of edges/links in the graph\\
    $N$ & Number of vertices/nodes ($N = |V|$)\\
    $M$ & Number of edges/links ($M = |E|$)\\
		$d_u$ & The degree of $u$\\
    $N(u)$ & The set of neighbors of $u$\\
    $T(u)$ & The number of triangles containing $u$\\
    $C(u),C(G)$ & Clustering coefficients of $u$ and $G$\\
    $\tilde C_v(u),\tilde C_v(G)$ & Clustering coefficients of $u$ and $G$\\
    & after removing node $v$ from $G$\\
    $G[S]$ & The subgraph induced by $S \subseteq V$ in $G$\\
    $tr(u,v)$ & The number of triangles containing both $u,v$\\
    \bottomrule
    \end{tabular}%
  \label{tab:syms}%
  %\vspace{-0.28in}
\end{table}% 
%%\vspace{-0.12in}
%----------------------%
%----------------------%
\subsubsection{Average Clustering Coefficient (ALCC)}
%----------------------%
In graph theory, the average local clustering coefficient (ALCC) 
$C(G)$  of a graph $G$ is a measure indicating how much vertices of $G$ tend to cluster together \citep{watts1998cds}. 
This measure is defined as the average of LCC over all vertices in the network. $C(G)$ is defined as:
%$C(G) = \displaystyle \frac{1}{N}\sum_{u \in V} C(u).$
%%%----------------------%
{ \color{black}
\begin{equation}
	C(G) = \frac{1}{N}\sum_{u \in V} C(u).
	\label{equ:def_gcc}
\end{equation} }
%%%----------------------%
Because $0\leq C(u) \leq 1$ for every node $u \in V$, $C(G)$ is normalized and can only take values in the range $[0, 1]$ inclusively. For instance, $C(G) = 0$ when $G$ is a triangle-free graph and $C(G) = 1$ when $G$ is a clique or a collection of cliques. 
{ \color{black} The higher the clustering coefficient of $G$ the more closely the graph
locally resembles a clique. Also, we define
\[ \tilde{C}_v(G) = C \left( G[ V \backslash \{ v \} ] \right). \] }
%As a result, ALCC is a widely-targeted measure 
%for complex network optimization problems \citep{Ponton2013}. 
%% A closely related definition for ALCC as the number of closed triplets
%% divided by the number of connected triplets of vertices; we do not use
%% this definition in this work.
%%The list of notations and their meanings is summarized in Table \ref{tab:syms}.
%----------------------%
\subsection{Problem definition}
%----------------------%
We define the \emph{Clustering Structure Assessment} problem (CSA) as follows
%----------------------%
\begin{definition}[$CSA(G, k)$]
\label{defCSA}
Given a network $G = (V, E)$ and a positive integer $k \leq N$, find a subset $S^* \subseteq V$ of cardinality at most $k$ that maximizes the reduction of the clustering coefficient, i.e.,
$$S^* = \displaystyle \underset{S \subseteq V, |S| \leq k}{\operatorname{argmax}} \Delta C(S),$$
where $\Delta C(S) = C(G) - C(G[V \backslash S]).$
\end{definition}
%----------------------%

CSA problem aims to identify the most critical
vertices of the network with respect to 
the average clustering coefficient.
The input parameter $k$ can be interpreted 
as the the maximum number of node 
failures that normal functionality of the network
can withstand once adversarial attacks or random corruptions occur. Accordingly, the case $|S| = k$ identifies exactly $k$ critical vertices and examines the worst scenarios that can happen when these vertices are { \color{black} compromised. }

{\color{black}
\subsection{Formulation as cubic integer program}
In this section, we formulate the CSA problem as
an integer program.
Let $(e_{ij})_{i,j \in V}$ be the adjacency matrix of $G$.
\begin{lemma}
  For $u \in V$, $T(u)$ can be calculated in the following way:
  \[ 2T(u) = \sum_{i \in V} \sum_{j \in V} e_{ui} e_{uj} e_{ij}. \]
\end{lemma}
\begin{proof}
  The summand $e_{ui}e_{uj}e_{ij} = 1$ iff $i,j$ are neighbors of $u$,
  and if edge $(i,j)$ is in the graph; that is, vertices $u, i, j$ form
  a triangle.
\end{proof}

We formulate CSA as an integer program in 
the following way. Let 
$x_i = 1$ if $i$ is included in the set $S$, 
and $x_i = 0$ otherwise.
\begin{IP}
\begin{equation} \label{eq:cubic} \min \sum_{u \in V : d(u) > 1} \sum_{i \in V} \sum_{j \in V} 
  \frac{ e_{ui}e_{uj}e_{ij} x_ix_jx_u }{d_u(d_u -1)(N-k)}
\end{equation}
such that
\[ \sum_{u \in V} x_u \le k, \]
\[ x_u \in \{0, 1 \}, \, u \in V. \]
\end{IP}
Notice that the sum (\ref{eq:cubic}) computes
the ALCC of the residual graph after removing
$S$. 
As we show 
in Section \ref{ssRandomFailures}, Corollary \ref{cor:dec}, 
there always exists a node the removal
of which will not increase the ALCC;
thus, an optimal solution to the program is an
optimal solution to CSA. }

\section{Complexity of CSA} \label{ssanalysisComplexity}
%%\vspace{-0.1in}

In this section, we show the NP-completeness of $CSA(G, k)$. This intractability indicates that an optimal solution for CSA might not be computationally feasible in practice.
%----------------------%
\begin{definition}[Decision problem -- $CSA(G, k, \alpha)$]
Given a network $G = (V, E)$, a number $k \leq N$ and a value $0 \leq \alpha \leq 1$, does there exist a set $S \subseteq V$ of size $k$ such that $\Delta C(G) \geq \alpha$?
\end{definition}
%----------------------%
%----------------------%
\begin{theorem}
$CSA(G, k, C(G))$ is NP-Complete.
\end{theorem}
%----------------------%
%----------------------%
%\begin{comment}
\begin{proof}
{ \color{black} We show that the following subproblem of $CSA(G, k, C(G))$ is NP-complete; the subproblem asks }
for a set $S \subseteq V$ of $k$ nodes whose removal completely degrades the clustering coefficient $C(G[V\backslash S])$ to 0, or equivalently, makes the residual graph $G[V\backslash S]$ triangle-free (Lemma \ref{pro:basic1}). To show the NP-completeness, we first show that CSA is in NP, and then prove its NP-hardness by constructing a polynomial time reduction from 3-SAT to $CSA(G, k, C(G))$.
Given a set $S \subseteq V$ of $k$ nodes, one can verify whether $G[V \backslash S]$ is triangle-free by computing the trace of $A^3$ where $A$ is the adjacency matrix of $G[V \backslash S]$. As we mentioned above, this can be done in $O((N-k)^{2.372})$. Therefore, $CSA(G, k, C(G))$ is in NP.

Now, given an instance boolean formula $\phi$ of 3-SAT with $m$ variables and $l$ clauses, we will construct an instance of $CSA(G, k, C(G))$, where $k = m + 2l$, as follows: { \color{black} 
\begin{enumerate}
 \item For each clause $C = l_1 \vee l_2 \vee l_3$ of $\phi$, introduce a 3-clique in $G$ with 3 clause literals as vertices: add vertices
   $l_1^C, l_2^C, \l_3^C$, and edges $(l_i^C, l_j^C)$ for $1 \le i < j \le 3$.
   Color these vertices blue. \label{step3clique}
 \item For each variable $x_i$ of $\phi$, create two vertices representing literals $x$ and $\neg x$ in $G$ and connect them by an edge. That is,
   add vertices $v_{x_i}, v_{\neg x_i}$ and edge $(v_{x_i}, v_{\neg x_i}$). Color these vertices green. \label{stepvariable}
 \item For each blue vertex in a 3-clique created in step \ref{step3clique}, connect it to the corresponding green literal created in step \ref{stepvariable}.
   That is, for each literal $l_j$ in each clause $C$, if $l_j = x_i$, then add edge $(l_j^C, v_{x_i})$. If literal $l_j = \neg x_i$, then
   instead add edge $(l_j^C, v_{\neg x_i})$. \label{stepjoin}
 \item Finally, for every edge in $G$, create a dummy vertex $d$ (color it red) and connect $d$ to the two endpoints of that edge.\label{stepfinal}
\end{enumerate}
%1) For each clause of $\phi$, introduce a 3-clique in $G$ with 3 clause literals as vertices. Color these vertices blue;
%2) For each variable $x$ of $\phi$, create two vertices representing literals $x$ and $\neg x$ in $G$ and connect them by an edge. Color these vertices green;
%3) For each blue vertex in a 3-clique created in step 1, connect it to the corresponding green literal created in step 2;
%4) Finally, for every edge in $G$, create a dummy vertex $d$ (color it red) and connect $d$ to the two endpoints of that edge.
%%%----------------------%
\begin{figure}[t]
  \centering
  \includegraphics[scale=0.23]{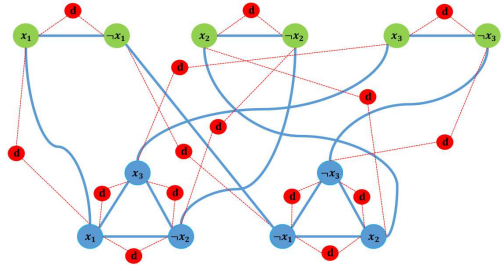}%
  \caption{{\color{black} Reduction example for a toy instance $(x_1 \vee \neg x_2 \vee x_3) \wedge (\neg x_1 \vee x_2 \vee \neg x_3)$ of 3-SAT.}}%
  \label{fig:npc}%
\end{figure}
%%%----------------------%

Figure \ref{fig:npc} illustrates 
the reduction of the toy boolean formula $(x_1 \vee \neg x_2 \vee x_3) \wedge (\neg x_1 \vee x_2 \vee \neg x_3)$. 
In this example, step \ref{step3clique} introduces two 3-cliques with blue vertices, step 2 creates three pairs of green vertices, 
and step \ref{stepjoin} consequently connects blue vertices to their corresponding green vertices by the thick curly edges. 
Finally, step \ref{stepfinal} assembles dummy nodes $d$'s (in red) and two dotted lines for every existing edges in $G$. }

Let $G_{-d}$ denote the graph $G$ without dummy vertices $d$'s and their adjacent dotted edges.
Assume that $\phi$ has a satisfied assignment, we construct $S$ by (i) include in $S$ all vertices corresponding to true literals, and (ii) for each clause, include in $S$ all vertices of the 3-clique but the one corresponding to its first true literal. Thus, $S$ includes $m$ green vertices and $2l$ blue vertices. It is verifiable that vertices in $S$ form the vertex cover of $G_{-d}$. As a result, the removal of all nodes in $S$ will make $G[V \backslash S]$ triangle-free (since it leaves no edges in $G_{-d}$).

Suppose there exists a set $S$ of $k$ nodes such that removing $k$ nodes in $S$ leaves $G[V \backslash S]$ triangle-free. We note that $S$ will not contain any dummy node $d$ because replacing $d$ by any of its adjacent literals (which are not already in $S$ yet) yields a better solution in term of triangle coverage. As a consequence, $S$ only contains blue and green vertices. Furthermore, 
nodes in $S$ have to be indeed the vertex cover of $G_{-d}$ in order for $G[V\backslash S]$ to be triangle-free. 
This cover must contain one green vertex for each variable and two blue vertices for each 3-clique (or clause), requiring exactly $k = m+2l$
vertices. Now, assign value \texttt{true} to the variables whose positive literals are in $S$. 
Because $k = m + 2l$, for each clause at least one edge connecting its blue 3-clique to the green vertices is covered by a variable vertex. 
Hence, the clause is satisfied.
\end{proof}

\section{Vulnerability Analysis in Random Failure} \label{ssRandomFailures}

\subsection{Monotonicity of ALCC}
%----------------------%
{ \color{black} The value of ALCC is not monotonic in terms of the set of excluded nodes $S$. }
Counterexamples showing the non-monotonicity of ALCC are presented in Fig. \ref{fig:monosub}. 
This implies that we do not always have either $C(G[V\backslash S_1]) \geq C(G[V \backslash S_2])$ or $C(G[V \backslash S_1]) \leq C(G[V \backslash S_2])$ for any subsets $S_1 \subseteq S_2 \subseteq V$. 
%So, how do we make sure that we can remove nodes to degrade ALCC of the residual network? Do such nodes always exist? 
{ \color{black} In fact, it is possible that ALCC could be at a local minimum with further node removal increasing
the value of ALCC.
Our analysis in Section \ref{sect:rf} shows that is always possible to degrade the value of 
ALCC by removing a vertex. 
We show that in any network $G$ there exists a vertex $u$ such that $\tilde C_u(G) \leq C(G)$. 
This result is the basis of the algorithms we present in Section \ref{ssnodeRemoval}.}
%-------------
%%\begin{figure}[t]%
	%%%\centering
	%%\includegraphics[scale=0.25]{figs/NONmonotoneExample.pdf}%	
	%%%\vspace{-0.1in}
	%%\caption{\small{Examples of the nonmonotonicity of ALCC. a) Removing the green vertex will reduce ALCC to 0 whereas b) Removing the green vertex will increase ALCC to 1.}}%
	%%\label{fig:mono}%
	%%%\vspace{-0.2in}
%%\end{figure}

\begin{figure}[t]%
\centering
	\subfigure[Nonmonotonicity of ALCC] {
		\includegraphics[scale=0.36]{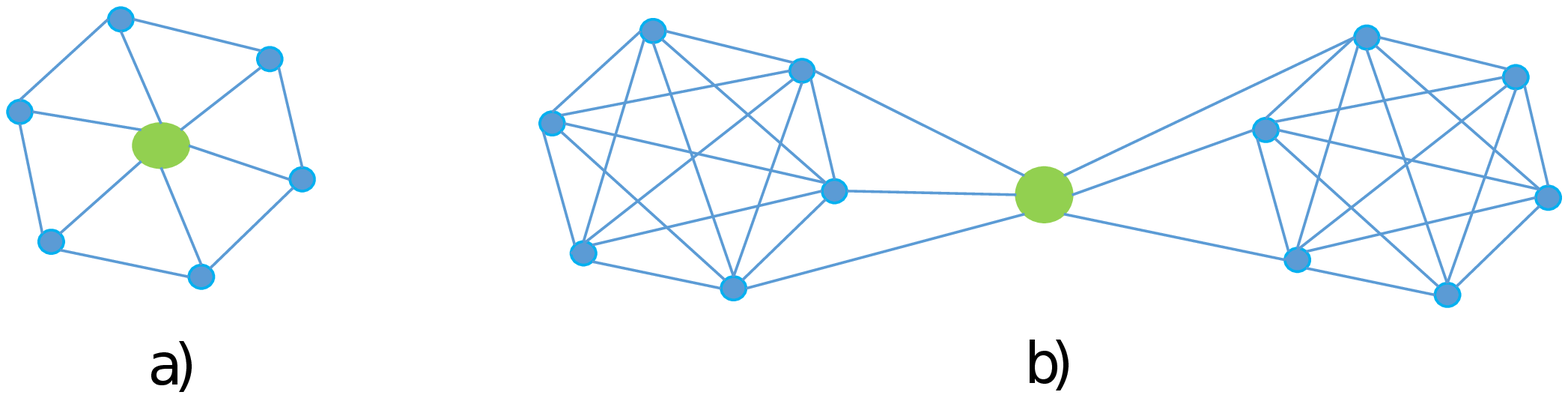}
		\label{fig:monosub}}

	%\vspace{-0.1in}
        \caption{Nonmonotonicity of ALCC. a) ALCC = 0 whereas b) ALCC = 1 when the green vertex is removed}%
	\label{fig:mono}%
	%\vspace{-0.23in}
\end{figure}

%----------------------%
%----------------------%
\subsection{Analysis { \color{black} of Random Failure}} \label{sect:rf}
When random failures occur, { \color{black} 
the ALCC value is unpredictable due to the nonmonotonicity of ALCC. }
That is, the removal of nodes can result in either higher or lower ALCC of the residual graph. 
We show that under uniform random failures the expected ALCC  
{ \color{black} $E_{u}[\tilde C_{u}(G)]$ } is at most the current ALCC value (Theorem \ref{theo:main1}). 
This result also indicates that, given a network $G$, there exists a sequence of 
subgraphs $G_i$ of $G$ whose ALCC values form a nonincreasing sequence (Corollary \ref{cor:frommain}). 

%----------------------%
\begin{lemma}
In a graph $G$, the following statements hold: %(i) $C(G) = 0$ iff $G$ is a triangle-free network, and (ii) $C(G) = 1$ iff $G$ is a clique or contains only separated cliques.
\begin{itemize}
  \item[(i)] $C(G) = 0$ if and only if $G$ is a triangle-free network.
  \item[(ii)] $C(G) = 1$ if and only if $G$ is a clique or contains only separated cliques.
\end{itemize}
\label{pro:basic1}
\end{lemma}
%----------------------%
%----------------------%
{ \color{black}
\begin{proof}
\begin{itemize}
\item[(i)] Suppose there exists a triangle $u,v,w$ in $G$. Then $C(u) > 0$, so
  $C(G) > 0$. For the converse, if $C(G) > 0$, there exists $u \in V$ such that
  $C(u) > 0$. By definition of $C(u)$, there exists a triangle $u,v,w$ containing $u$.
\item[(ii)]
  Suppose $C(G) = 1$. Then, for each $u \in V$, the subgraph induced by $\{ u \} \cup N(u)$ is a clique,
  from which $G$ is a clique or only separated cliques. The converse follows directly from the
  definition of $C(G)$.
  \end{itemize}
\end{proof} }
%----------------------%
%----------------------%
\begin{lemma}
For any $u \in V$,
$$2T(u) = \displaystyle\sum_{v \in N(u)} |N(u) \cap N(v)|.$$
\label{lem:triangle}
\end{lemma}
%----------------------%
{ \color{black} 
\begin{proof}
For each neighbor $v$ of $u$, the number of triangles that contain both $u$ and $v$ is $|N(u) \cap N(v)|$.  
Since each triangle containing $u$ contains exactly two neighbors of $u$, it follows that the summation 
$\sum_{v \in N(u)} |N(u) \cap N(v)|$ counts twice the number of triangles containing $u$. 
\end{proof} }
%----------------------%
%----------------------%
\begin{lemma}
For any node $u \in V$:  
	\begin{equation}
		\frac{1}{N-1}\sum_{v \in V \setminus \{u\}} \tilde C_v(u) \leq  C(u).
		\label{eq:removalMain}
	\end{equation}
	\label{lem:lcc_pro}
\end{lemma}
%----------------------%
%----------------------%
\begin{proof}
To prove this Lemma, we will show the following statements regarding the degree of $u$:
\begin{align}
	\frac{1}{N-1}\sum_{v \in V \setminus \{u\}} \tilde C_v(u) \leq  C(u) \text{ \quad when $d_u \leq 2$}.	
	\label{eq:removal2}
\end{align}
%----------------------%
%----------------------%
\begin{equation}
	\frac{1}{N-1}\sum_{v \in V \setminus \{u\}} \tilde C_v(u) =  C(u)	\text{ \quad when $d_u > 2$}.
	\label{eq:removal}
\end{equation}
%----------------------%
Eq. (\ref{eq:removalMain}) is equivalent to 
$ \displaystyle\sum_{v \in V\setminus \{u\}} \tilde C_v(u) \leq (N-1) C(u).$

Expanding the left-hand-side (LHS) of this inequality yields
\begin{align}
\nonumber \sum_{v \in V\setminus \{u\}} \tilde C_v(u) &= \sum_{v \in N(u)} \tilde C_v(u)  + \sum_{v \in V\setminus (N(u)\cup \{u\}) } \tilde C_v(u) \\
&= \sum_{v \in N(u)} \tilde C_v(u) +  (N-d_u - 1) C(u).
\label{eq:total}
\end{align}

To find $\sum_{v \in V\setminus (N(u)\cup \{u\}) } \tilde C_v(u)$, we use the fact that removing a non-neighbor node of $u$ will not affect the local clustering coefficient $C(u)$, i.e., $\tilde C_v(u) = C(u)$ for $v \in V\setminus (N(u)\cup \{u\})$. There are $(N - d_u - 1)$ non-neighbors vertices of $u$ in $G$. Thus the second term of (\ref{eq:total}) follows. To evaluate the first term of Eq. (\ref{eq:total}), we consider two cases:
 
Case (i): When $d_u \leq 2$ (i.e., $u$ has only one or two neighbors). In this case, the removal of any neighbor of $u$ will make $d_u \leq 1$, and thus, will drop $\tilde C_v(u)$ to 0 based on the definition of LCC. This implies
$0 =  \displaystyle\sum_{v \in N(u)} \tilde C_v(u) \leq d_u \times C(u).$
Substituting this to the first term of Eq. (\ref{eq:total}) yields Eq. (\ref{eq:removal2})

Case (ii): When $d_u > 2$. For any $v \in N(u)$, removing $v$ degrades $d_u$ to $d_u - 1$ and decreases the number of triangles on $u$ by an amount of $|N(u) \cap N(v)|$. As a result,
\begin{equation}
%\tilde C_v(u) = \frac{2 \left(T(u) - |N(u) \cap N(v)|\right)}{ (d_u - 1) (d_u - 2)}, \quad \forall v \in N(u).
\tilde C_v(u) = \frac{2 \left(T(u) - |N(u) \cap N(v)|\right)}{ (d_u - 1) (d_u - 2)}.
\end{equation}

Therefore,
\begin{align}
\nonumber\sum_{v \in N(u)} \tilde C_v(u) = \frac{\sum_{v \in N(u)}2 \left(T(u) - |N(u) \cap N(v)|\right)}{ (d_u - 1) (d_u - 2)}\\
\label{eq:tu}= \frac{ 2 (d_u T(u) - \sum_{v\in N(u)}|N(u) \cap N(v)|)}{(d_u - 1) (d_u - 2)}
\end{align} 

By Lemma \ref{lem:triangle}, we can simplify Eq. (\ref{eq:tu}) to
\begin{align*}
 \sum_{v \in N(u)} \tilde C_v(u)=\frac{  2(d_u-2) T(u)}{(d_u - 1) (d_u - 2)} = d_u C(u)
\end{align*} 

Substituting this to Eq.~(\ref{eq:total}) yields Eq. (\ref{eq:removal}). The \emph{inequality} in Lemma \ref{lem:lcc_pro} occurs only if $u$ is of single degree, or $u$ has exactly two connected neighbors.
\end{proof}
%----------------------%

Using Lemma \ref{lem:lcc_pro}, we can show the following main result of ALCC's behavior on random failures:
%----------------------%
\begin{theorem}
In a graph $G$, { \color{black} $E_u[\tilde C_{u}(G)] \leq C(G)$}.
\label{theo:main1}
\end{theorem}
%----------------------%
%----------------------%
\begin{proof}
By definition of ALCC, we have
\begin{align*}
 \sum_{v \in V}\tilde C_v(G) &= \sum_{v \in V}\left[ \frac{1}{N-1} \sum_{u \in V\setminus \{v\}}\tilde C_v(u) \right].
\end{align*}

Applying Eq. (\ref{eq:removalMain}) in Lemma \ref{lem:lcc_pro} gives
%----------------------%
\begin{align*}
 \sum_{u \in V}\left[ \frac{1}{N-1} \sum_{v \in V \setminus\{u\}}\tilde C_v(u) \right] &\leq \sum_{u \in V} C(u) = N \times C(G).
\end{align*}
Thus $E[\tilde C_{\cdot}(G)] =  \displaystyle\frac{1}{N}\sum_{v \in V}\tilde C_v(G) \leq C(G).$
\end{proof}
%----------------------%
%----------------------%
\begin{corollary} \label{cor:dec}
In a graph $G \equiv G_0$ of $N$ nodes, there exists a sequence of subgraphs $G_0 \supseteq G_1 \supseteq \cdots \supseteq G_N \equiv \emptyset$ such that $C(G_i) \geq C(G_{i+1})$ and $G_{i+1}$ is constructed by removing one vertex from $G_i$ for $i = 0, \ldots, N-1$.
\label{cor:frommain}
\end{corollary}

\section{Algorithms} \label{ssnodeRemoval}

In this section, we present two algorithms for CSA problem, namely \texttt{simple\_greedy} (Alg. \ref{algo:greedy}), and Fast Adaptive Greedy Algorithm (FAGA) (Alg. \ref{algo:fast_greedy}). 
Alg. \ref{algo:greedy} is a simpler greedy algorithm than FAGA,
which employs more sophisticated strategies { \color{black} to more efficiently 
provide a solution of significantly higher quality than Alg. } \ref{algo:greedy}.
%----------------------%
\begin{algorithm}
  \caption{Greedy Algorithm (\texttt{simple\_greedy})}
    \label{algo:greedy}
  \begin{algorithmic}[1]
    \State $S \gets \emptyset$;
    \For{  each $u \in V$ }
			\State $\tilde C_{u}(G) \gets C(G[V\setminus\{u\}]$);
    \EndFor
    \State $S \gets k \textrm{ vertices with lowest } \tilde C_{\cdot}(G)$ values;
    \State \Return{$S$}
  \end{algorithmic}
\end{algorithm}
%----------------------%
%\vspace{-0.2in}
\subsection{Simple Greedy Algorithm}
Our first algorithm (Alg.~\ref{algo:greedy}) computes for each node $u$ the ALCC value after removing $u$, denoted by $\tilde C_{u}(G)$. The $k$ vertices associated with the lowest values of $\tilde C_{u}(G)$ are included in the solution. Notice that in this algorithm, the values of $\tilde C_{u}(G)$ are computed only once, and $k$ nodes are simultaneously included in the final solution. 
{ \color{black} Since the local clustering coefficient of a node $u$ is 
dependent only on the subgraph of its neighbors, we chose this approach over 
iteratively recomputing
$\tilde C_{u}( G \backslash \{ s_1, \ldots, s_i \} )$ for all nodes after
choosing $\{ s_1, \ldots, s_i \}$ into set $S$. }

\textit{Time-complexity}:
The complexity of Alg. \ref{algo:greedy} depends on the $N$ calls to compute the ALCC of the network. There are two state-of-the-art methods in \citep{Gall14} and \citep{Schank05} for this purpose. 
If ALCC is computed using the matrix multiplying technique in \citep{Gall14}, 
the time-complexity is $O(N^{\omega})$ with $\omega \leq 2.372$. Alternatively, if ALCC is computed 
using the method in \citep{Schank05}, which has complexity of $O( M^{3/2})$, the overall complexity will be $O( N M^{3/2})$. 
In practice, neither of these two upper bounds fully dominates the other. 
{ \color{black} In our experimental evaluation in Section \ref{ssexperiment}, 
we utilize \citep{Schank05} for computing ALCC. }

\subsection{Fast Adaptive Greedy Algorithm}
%%\vspace{-0.07in}
We next present the Fast Adaptive Greedy Algorithm (FAGA - Alg. \ref{algo:fast_greedy}) 
that significantly improves \texttt{simple\_greedy}. 
For small values of $k$, this algorithm requires as much time as computing ALCC only once; 
it is $N$ times faster than its predecessor. 
Furthermore, it provides a significant quality improvement over \texttt{simple\_greedy} 
in our empirical studies.

In principle, FAGA employs an adaptive strategy in computing the reduction of ALCC when nodes are removed iteratively. At each round, the node $v$ incurring the highest reduction in ALCC is selected into the solution. As shown in the proof of Theorem \ref{theo:main2}, a node $v$
does exist at any iteration. Node $v$ is removed from the graph and the 
procedure repeats itself for the remaining vertices; that is, FAGA recomputes 
for each vertex $u$, which is not yet in the solution, its ALCC reduction $\Delta\tilde C_{u}$ 
when $u$ is removed from the graph. 
This strategy provides better solution quality than the non-adaptive greedy algorithm. 
While it is more complicated than the previous approach, 
it can be done faster than \texttt{simple\_greedy} as we show in the following discussion.

We structure FAGA into two phases. The first phase (lines 1--15) extends the algorithm in \citep{Schank05} to compute both ALCC and the number of triangles that are incident with each edge and node in the graph. This algorithm was proved to be time-optimal in $\theta(M^{3/2})$ for triangle-listing, and has been shown to be very efficient in practice.
The second phase (lines 16--33) repeats the vertex selection for $k$ rounds. In each round, we select the node $u_{max}$ { \color{black} which decreases the clustering coefficient }
the most into the solution, remove $u_{max}$ from the graph, and perform the necessary update for $\Delta\tilde C_{u}$ for the remaining nodes $u \in V$.

The key efficiency of FAGA algorithm is in its update procedure for 
$\Delta\tilde C_{u}$. The update $\Delta\tilde C_{u}$ for remaining nodes 
after removing $u_{max}$ can be done in linear time. 
This is made possible due to the information on the number of triangles involving each edge. The correctness of this update formulation (lines 18--26) is proved in the following lemma.

%----------------------%
\begin{algorithm}[t]
  \caption{Fast Adaptive Greedy Algorithm (FAGA - \texttt{fast\_greedy})}
    \label{algo:fast_greedy}
  \begin{algorithmic}[1]
      \State Number the vertices from $1$ to $N$ such that $u < v$ implies $d(u) \leq d(v)$.
      \State $S \gets \emptyset$;
      \For{ each $u \in V$ }
      	 $T(u) \gets 0$;
      \EndFor
      \For{ each $(u, v) \in E$ }
      	 $tr(u, v) \gets 0$;
      \EndFor      
      
      \For{$u \gets n \textrm{ to } 1$}
       		\For{ \textrm{each } $v \in N(u)$ with $v<u$ }
       			\For{\textrm{each } $w \in A(u) \cap A(v)$}
       			\State Increase $tr(u, v), tr(v, w)$ and $tr(u, w)$ by one;
       			\State Increase $T(u), T(v)$ and $T(w)$ by  one;
       			\State Add $u$ to $A(v)$;
       			\EndFor	
       		\EndFor
      \EndFor            
      \For{  $i\gets 1$ to $k$ }
      	\For{  each $u \in V \setminus S$ }
      		\State { \color{black} $\Delta\tilde C_{u} \gets \frac{2T(u)}{Nd(u)(d(u) - 1)}$;}
      		\For{ each $v \in N(v) \setminus S$ }
                \If{ $d(v) > 2$ }
      			\State { \color{black} $\Delta\tilde C_{u} \gets \Delta\tilde C_{u} + \frac{4T(v)(1 - N) + 2tr(u,v)Nd(v) - 2T(v)d(v)}{N(N-1)d(v)(d(v)-1)(d(v)-2) } $;}
                \EndIf
                \If{ $d(v) = 2$ }
      			\State { \color{black} $\Delta\tilde C_{u} \gets \Delta\tilde C_{u} + T(v) / N$; }
                 \EndIf
      		\EndFor
      	\EndFor      	 
      	\State $u_{max} \gets \arg\max_{u \in V \setminus S}\{ \Delta\tilde C_{u} \}$;
      	\State Remove $u_{max}$ from $G$,
        add $u_{max}$ to $S$, and decrease $N$ by one;      	
      	\For{ each $ (v, w) \in E$ and $ v, w \in N(u_{max}) \setminus S$ }
      				\State Decrease $T(v)$ and $T(w)$ by one;
      	\EndFor      	
      \EndFor      
      \State \Return{$S$}			
  \end{algorithmic}	
\end{algorithm}
%----------------------%
%----------------------%
\begin{lemma} {\color{black}
Let $N_2(u) = \{ v \in N(u): d(v) = 2 \}$, 
$N_{>2}(u) = \{v \in N(u) : d(v) > 2 \}$.
%and $\Delta\tilde B_{u}(G) = \sum_{v \in G} C_u(v) - \sum_{v \in G} C(v).$
%Maximizing the difference in clustering coefficient over
%removal of a single node is equivalent to maximizing
%\[ \max_{u \in V} \Delta\tilde B_{u}(G). \]
For each $u \in V$, $\Delta\tilde C_{u}(G)$ can 
be computed in the following way:
\begin{align*}
\Delta\tilde C_{u} = \frac{2T(u)}{Nd(u)(d(u) - 1)} +& \sum_{v \in N_{>2}(u)}\frac{4T(v)(1 - N) + 2tr(u,v)Nd(v) - 2T(v)d(v)}{N(N-1)d(v)(d(v)-1)(d(v)-2) } \\ 
+& \sum_{v \in N_2(u)} \frac{T(v)}{N}
\end{align*}}
\end{lemma}
%----------------------%
%----------------------%
\begin{proof}
Denote the contribution of $v \in G$ to the average clustering coefficient as $c_v$ before the removal of $u$
and $\hat{c}_v$ after.
$\Delta\tilde C_{u}$ can be written as 
$\sum_{v \in G} c_v - \hat{c}_v.$
If $v \not \in N(u) \cup \{ u \}$, then $c_v = \hat{c}_v$. 
If $v = u$, then {\color{black}
$$c_v - \hat{c}_v = \frac{2T(u)}{Nd(u)(d(u) - 1)}.$$}
Let $v \in N_{>2}(u)$.  Then before removal of $u$, $v$ is in $T(v)$ triangles.  After removal,
$v$ is in $T(v) - tr(u,v)$ triangles.
Hence {\color{black}
$$c_v = \frac{2T(v)}{Nd(v)(d(v) - 1)},$$}
and {\color{black}
$$\hat{c}_v = \frac{2(T(v) - tr(u,v))}{(N-1)(d(v) - 1)(d(v) - 2)},$$} whence {\color{black}
\[c_v - \hat{c}_v = \sum_{v \in N_{>2}(u)}\frac{4T(v)(1 - N) + 2tr(u,v)Nd(v) - 2T(v)d(v)}{N(N-1)d(v)(d(v)-1)(d(v)-2) }.\]}
Let $v \in N_{2}(u)$. Before removal of $u$, $v$ is in $T(v)$ triangles. 
After removal, $v$ is in 0 triangles, hence the result follows.
\end{proof}

One important feature of FAGA is that the produced residual ALCC values will form a nonincreasing sequence. 
This feature is summarized in the following theorem.
%--------------------------%
%--------------------------%
%----------------------%
\begin{theorem}
The ALCC values of networks after each iteration (Alg.~\ref{algo:fast_greedy}, lines 16 -- 28) form a non-increasing sequence.
\label{theo:main2}
\end{theorem}
%----------------------%
%%\vspace{-0.1in}
\begin{proof}
We first show that in a graph $G$, there always exists a node $u$ such that $\tilde C_u(G) \leq C(G)$. Assume otherwise, that is $\tilde C_v(G) > C(G)$ for all node $v \in V$. This implies $\sum_{v \in V}\tilde C_v(G) > N \times C(G)$ which contradicts Theorem \ref{theo:main1}. Thus, the statement holds true. Finally, the theorem follows because at each step we select the nodes that maximally degrades ALCC of the whole network.
\end{proof}
%----------------------%
\textit{Time-complexity}: The first phase takes $O(M^{3/2})$ as in \citep{Schank05}. The second phase takes a linear time in each round and has a total time complexity $O(k (N+M) )$. Thus, the overall complexity is $O(M^{3/2} + k (M+N) )$. When $k < M^{1/2}$, the algorithm has an effective time-complexity $O(N^{3/2})$, which is $N$ times faster than \texttt{simple\_greedy}.

{ \color{black} \section{Clustering and the spread of information} \label{sect:rel_gcc}
In this section, we provide additional evidence
for the relationship between the propagation of 
information in a social network and the average
network clustering.  Since information cannot propagate
from one connected component to another, we consider
this relationship when the graph $G$ representing
the social network is connected. Thus, we consider
connected graphs with different values of ALCC. 
We define the relevant models of influence propagation
in Section \ref{sect:rel_mod};
then, we demonstrate
an empirical relationship in Section \ref{sect:rel_exp}; 
next, we provide theoretical evidence in support of
this relationship
in Section \ref{sect:rel_theo}.
\subsection{Models of influence} \label{sect:rel_mod}
To observe the effect of ALCC on influence propagation,
we adopted the following two standard models \citep{Kempe2003};
intuitively, the idea of a model 
of influence propagation in a network is
a way by which nodes can be activated 
given a set of seed nodes. 
%In this work, we 
%use $\sigma$ to denote a model of influence propagation.
%Such a model is usually probabilistic, and
%the notation $\sigma(S)$ will denote the
%expected number of activations under the model $\sigma$
%given seed set $S \subseteq V$.  
An instance 
of influence propagation on a graph $G$ follows the independent cascade (IC) 
model if a weight can be assigned
to each edge such that the propagation probabilities can be computed as
follows: once a node $u$ first becomes active, 
it is given a single chance to activate each currently inactive neighbor $v$ with probability proportional to 
the weight of the edge $(u,v)$. 
In the linear threshold (LT) model each network user $u$ has an associated threshold $\theta(u)$ chosen uniformly from $[0,1]$ which determines how much influence (the sum of the weights of incoming edges) is required to activate $u$. $u$ becomes active if the total influence from its active
neighbors exceeds the threshold $\theta(u)$.
\subsection{Experimental evidence} \label{sect:rel_exp}
To test the relationship between influence 
propagation and clustering empirically, we
used a variety of Watts-Strogatz graphs \citep{watts1998cds};
a graph generated by this model starts as a ring
lattice, defined as follows. First, $n$ circular rings are constructed: for 
each $j \in \{ 1, \ldots, n \}$, vertices
$u_1^j, \ldots, u_n^j$ and edges $(u_{i}^j, u_{i + 1}^j)$, $i = 1, \ldots, n - 1$,
and $(u_n^j, u_1^j)$. Next, add edges $(u_{i}^{j}, u_i^{j + 1})$, for $j = 1, \ldots, n-1$,
and $(u_i^{n}, u_i^1)$, for each $i$. Finally, all vertices within $k$ hops of each
other are connected by an edge. For these experiments, we used $n = 100$ and $k = 3$.
With probability $p$, each edge
in the graph is rewired; that is, replaced
with an edge between two uniformly randomly
chosen vertices. By varying $p$, one can
control the level of clustering in the
network, as shown in
Fig. \ref{fig:rel_gcc}. Each graph generated
in this manner has the same number of edges.

The expected activation was computed
using a single seed node and an IT or LT realization; this
computation was averaged over 1000 trials. When we normalize by
the initial value, Fig. \ref{fig:rel_gcc} shows a remarkable similarity
between the normalized ALCC value and the normalized activations, for
both the IC and LT models. Therefore, these results provide evidence
supporting a positive correlation between ALCC and 
the expected activation of both the IC and LT models of information propagation.

\subsection{Theoretical evidence of relationship between ALCC and influence propagation} \label{sect:rel_theo}
In this section, we provide further
evidence supporting
the relationship between clustering
and influence propagation, in the form of the 
following proposition, which shows how the
probability of activation increases when 
more neighbors are shared; with higher ALCC,
we may expect a higher fraction of shared neighbors
between adjacent nodes.

\begin{proposition} Suppose $s$ is
activated; let $t$ be a neighbor of $s$,
and suppose $s,t$ share $k$ neighbors.
Consider the IC model with 
uniform probability $1 / 2$ on 
each edge. Then 
\[ Pr \left( \text{ $t$ becomes activated } \right) \ge 1 - (1/2) \cdot (3/4)^k . \]
\end{proposition}
\begin{proof}
  Let $A$ be the event that 
  edge $(s,t)$ exists, and let
  $B$ be the event that
  edge $(s,t)$ does not exist,
  but for a common neighbor $n$,
  the edges $(s,n),(n,t)$ exist.
  For each common neighbor $n$,
  let $A_n$ be the event 
  that both edges
  $(s,n),(n,t)$ exist. Then
  \begin{align*} 
    Pr \left( \text{ $t$ becomes activated } \right)
    &\ge Pr(A) + Pr(B) \\
&= 1/2 + 1/2 \cdot Pr \left( \bigcup_{n \in N(s) \cap N(t)}  A_n  \right) . 
  \end{align*}

  Notice that $Pr ( A_n ) = 1/4$, and
  let $N(s) \cap N(t) = \{ n_1, \ldots, n_k \}$.
  By the inclusion-exclusion principle,
  we have that 
  \[ Pr \left( \bigcup_{i=1}^k A_{n_{k}} \right) =
  \sum_{i = 1}^k { {k} \choose {i} } 
  (1/4)^i (-1)^{i + 1} = 1 - (3/4)^k. \]
  Therefore, 
  $Pr(A) + Pr(B) = 1 - (1/2) \cdot (3/4)^k$.

\end{proof} }

\section{Experimental evaluation} \label{ssexperiment}

%%\vspace{-0.1in}
We present the empirical results of our proposed algorithms on synthesized and real networks. { \color{black}
In Section \ref{sect:exp_data}, we describe our methodology; in Section \ref{sect:exp_acc},
\ref{sect:exp_lcc} we analyze the 
efficacy of degrading ACC, LCC, respectively;
in Section \ref{sect:exp_rt}, we analyze the running
time of the algorithms.}

%All methods are executed with update after each node removal.
%--------------------------%
{ \color{black}
\subsection{Methodology} 
\label{sect:exp_data}
\paragraph{Algorithms}
We are unaware of any competitive method that specifically minimizes ALCC,
so to evaluate our approaches we compare to the following strategies: 
\begin{itemize}
\item \texttt{random\_fail}: Remove nodes uniformly at random, 
\item \texttt{lcc\_greedy}: Remove nodes in greedy fashion 
according to highest local clustering coefficient, 
\item \texttt{max\_degree}: Remove nodes in greedy fashion 
according to highest degree, 
\item \texttt{betweenness}: Removes in greedy fashion 
according to the highest betweenness centrality. 
\item \texttt{optimal}: For a network with 35 nodes,
  we were able to compute the optimal solution to 
  CVA by exhaustive enumeration.
\end{itemize}
Method legends are described in Fig. \ref{fig:legend}.}
\paragraph{Datasets}
We use Erd\H{o}s-R\'enyi (ER) \citep{Erdos60onthe}, 
Watts-Strogatz (WS) \citep{watts1998cds}, 
and Barabasi-Albert (BA)\citep
{Albert2002} models to generate synthesized testbeds. 
These are foundational models which have been widely used in the literature.
%% ; we
%% briefly describe the WS model:
%% the WS model begins from a regular grid or lattice; with some probability $p$ each
%% edge in the graph is rewired to a random edge.
We used the following parameter values:
$N=10000 ,M=49772$, $p = 0.001$ (ER model); 
$N = 35$, $p = 0.2$ (ER model);
$N=15000, M=44994$ (BA model); 
and $N=10000, M=200000$, {\color{black} with $n = 100$, $k = 3$, and $p = 0.3$, these
parameters are defined in Section \ref{sect:rel_gcc} (WS model).}

Real-world traces include Facebook \citep{facebookdataset}, ArXiv ePrint citation \citep{arxivdataset}, and NetHEPT networks \citep{chen10kdd}. 
The trace of Facebook has 25,492 users and 
464,237 friendship links, 
NetHEPT has 15,234 authors
with 31,376 connections, and 
ArXiv has 26,197 nodes with 14,484 edges.
The parameter $k$ is set to a fraction of the total number of nodes in each graph.
Besides ALCC, we also evaluate how the removal of
critical nodes affects the maximum Local Clustering Coefficient (LCC).

%--------------------------%
\subsection{Results on Average Clustering Coefficient} \label{sect:exp_acc}

\begin{figure*}
  \centering
	\subfigure[Erd\H{o}s-R\'enyi] {
		\includegraphics[width=0.4\textwidth]{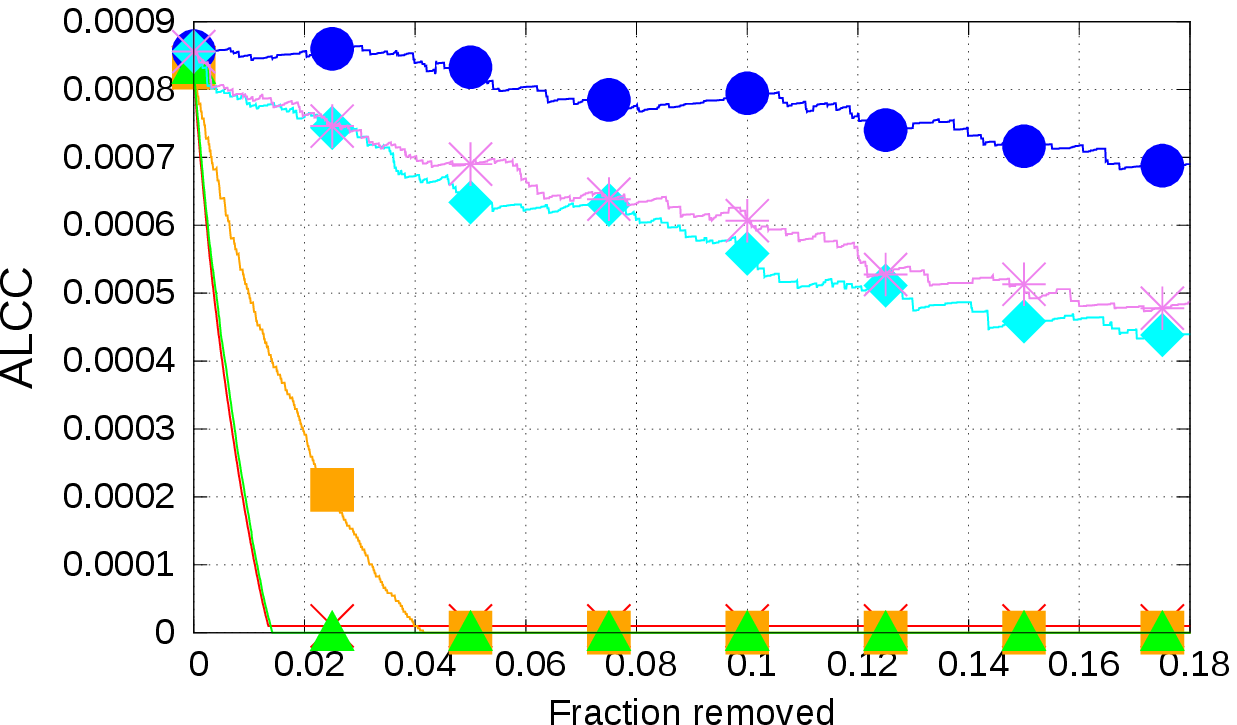}
		\label{fig:gccER}
	}  
  \subfigure[Watz-Strogatz] {
		\includegraphics[width=0.4\textwidth]{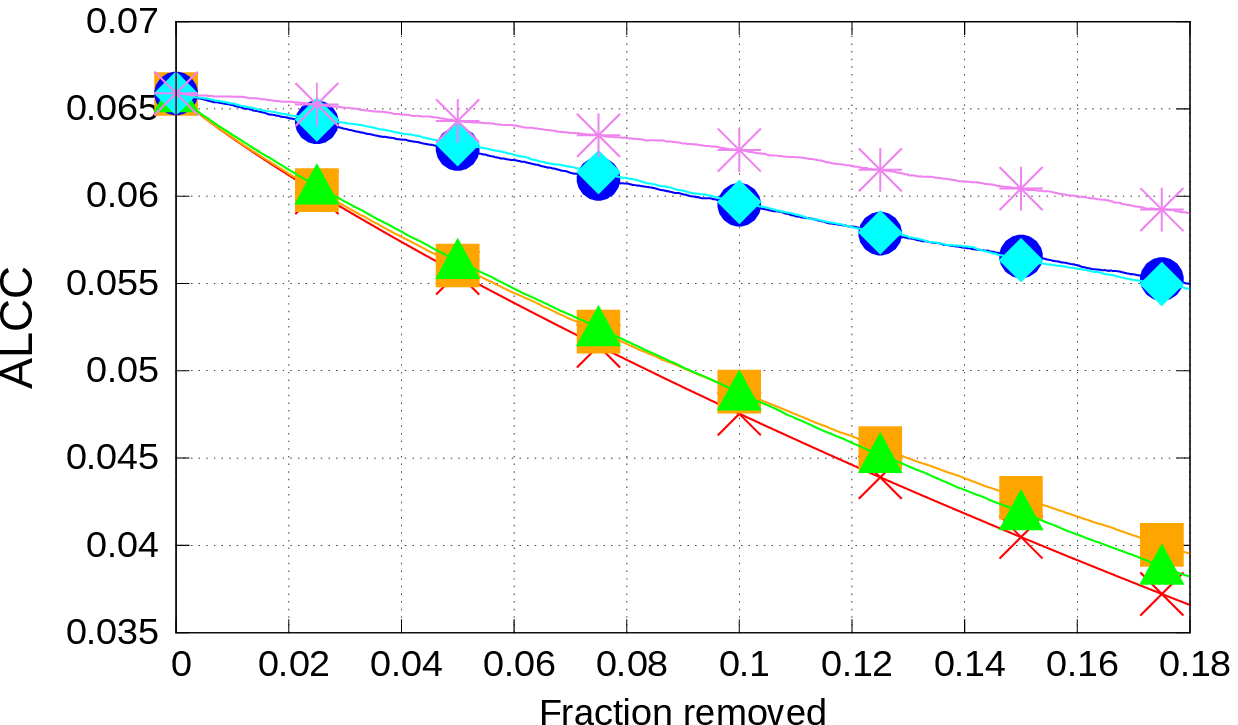}
		\label{fig:gccWS}
	}
  
  \subfigure[Barab\'asi-Albert] {
		\includegraphics[width=0.4\textwidth]{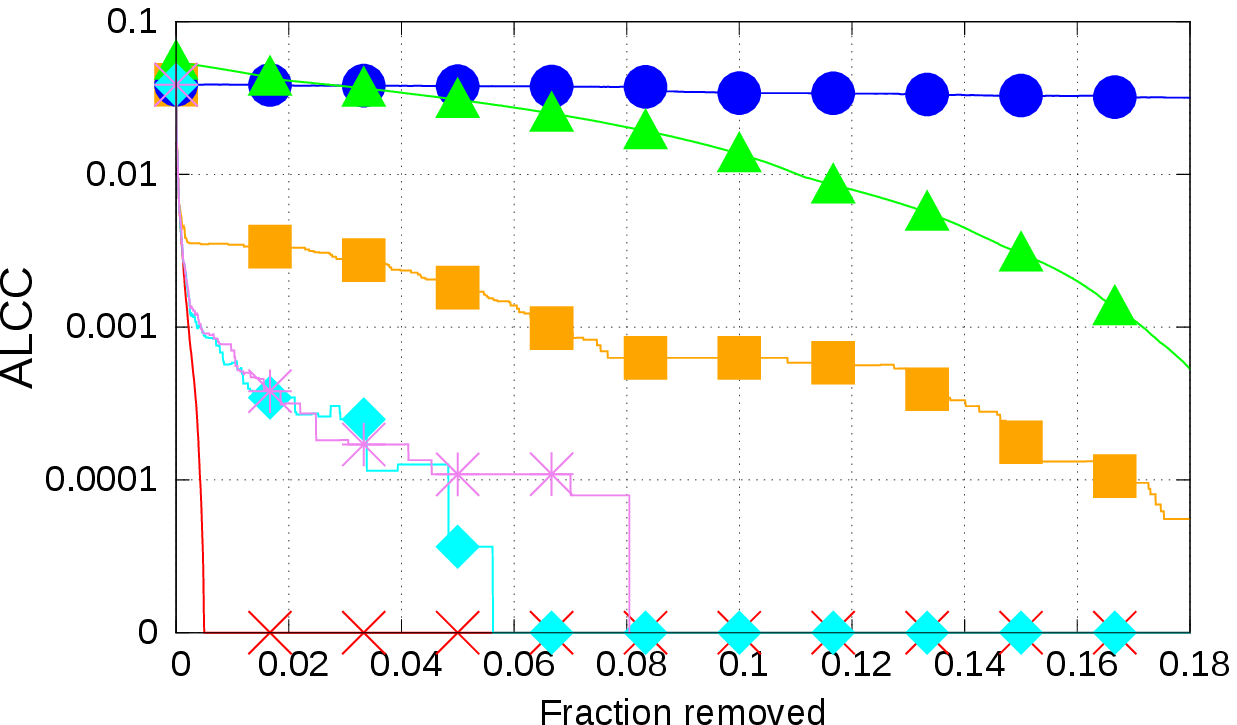}
		\label{fig:gccBA}
	}
	\subfigure[Arxiv] {
		\includegraphics[width=0.4\textwidth]{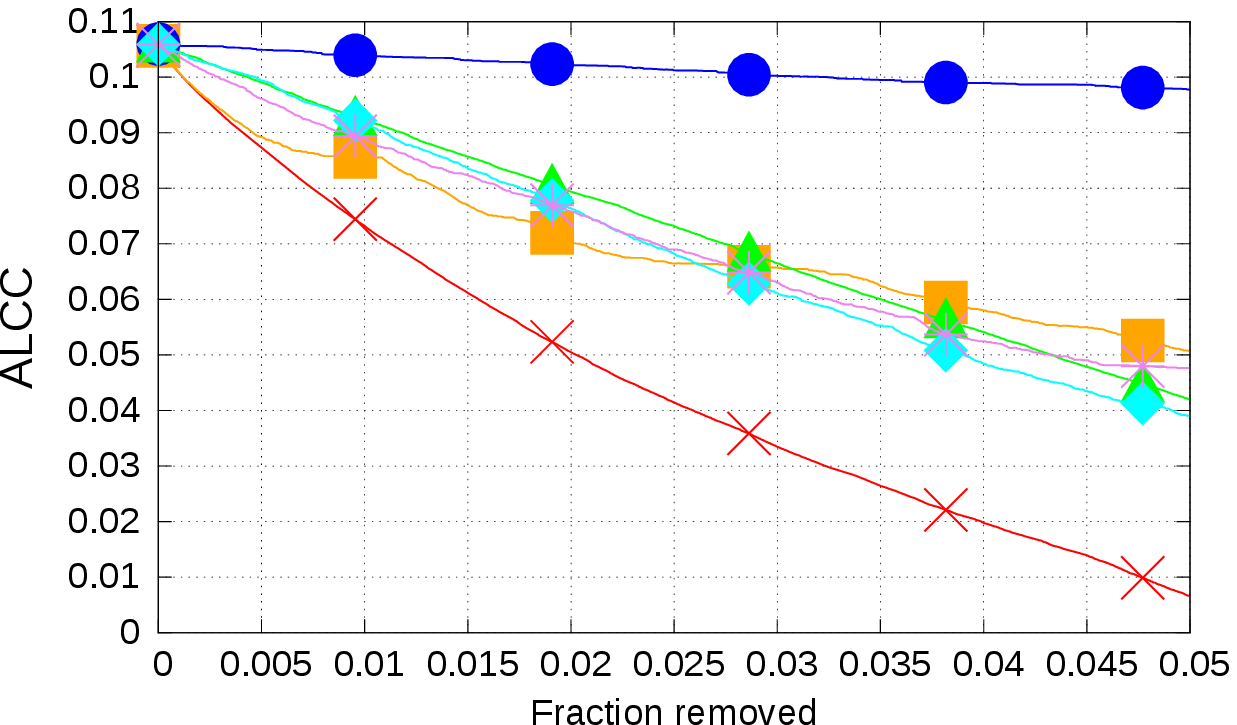}
		\label{fig:gccArxiv}
	}  

  \subfigure[NetHEPT] {
		\includegraphics[width=0.4\textwidth]{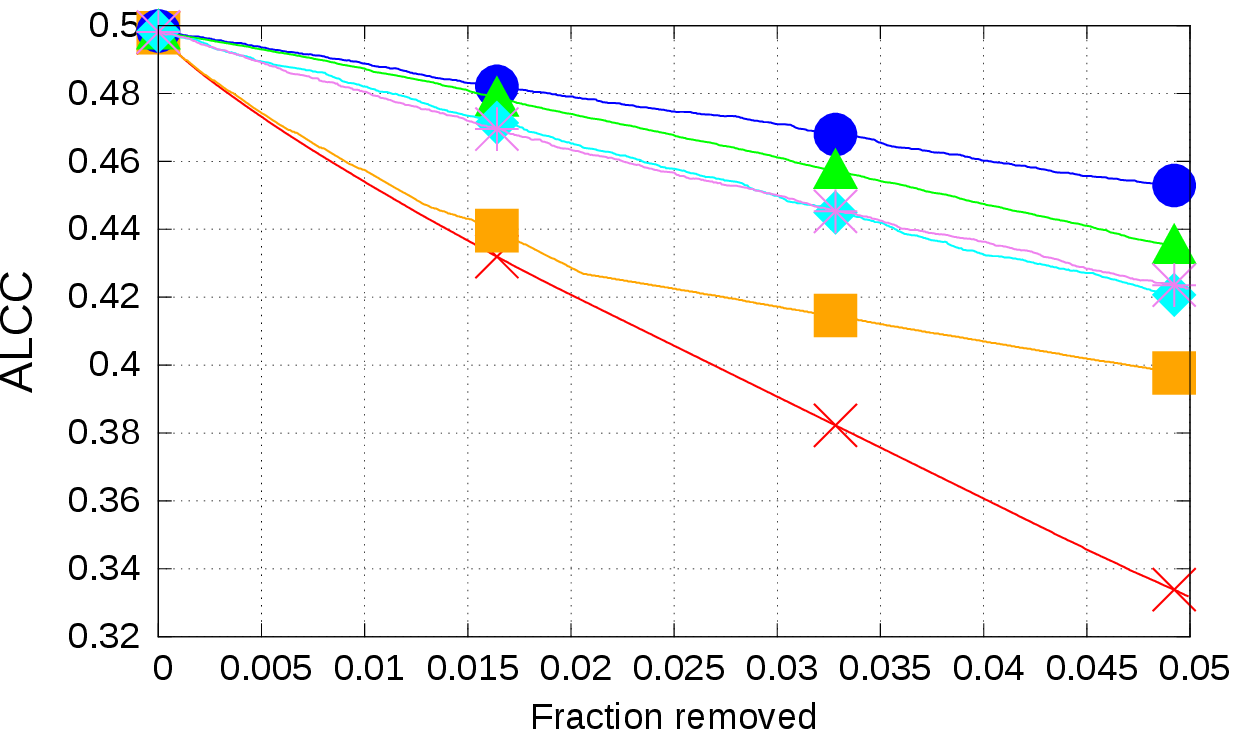}
		\label{fig:gccNET}
	}
  \subfigure[Facebook] {
		\includegraphics[width=0.4\textwidth]{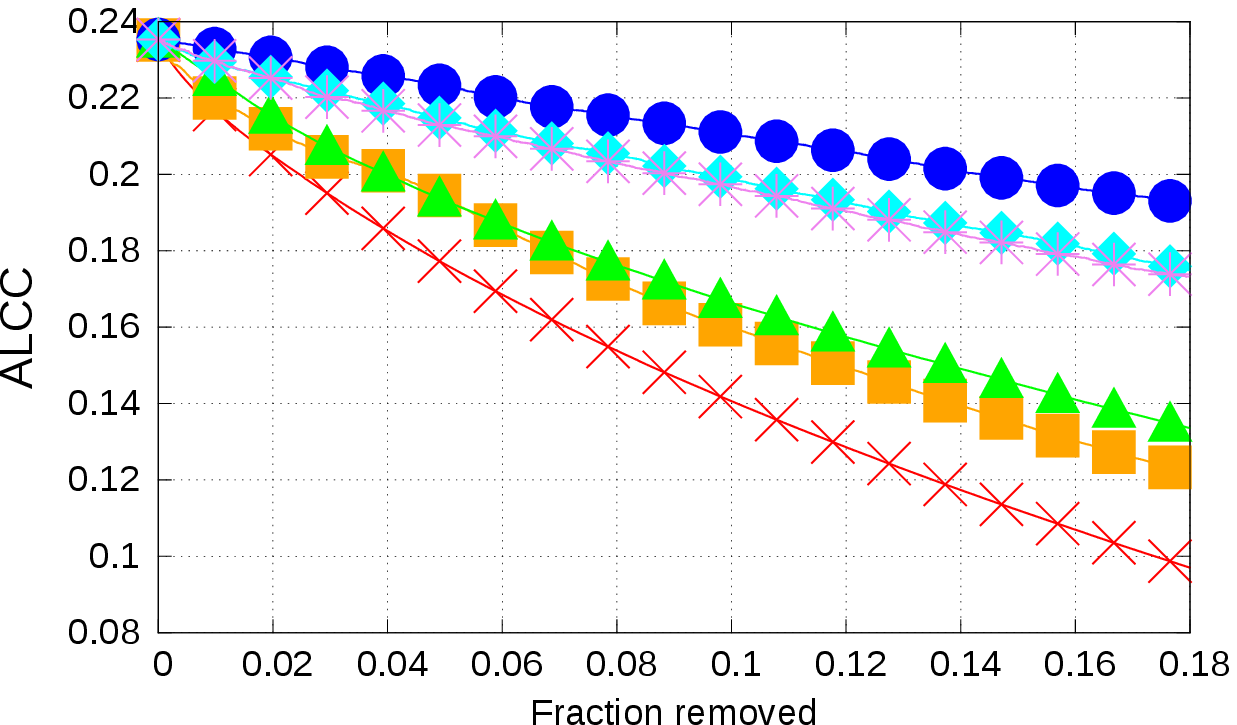}
		\label{fig:gccFB}
	}

  \subfigure[{\color{black}Erd\H{o}s-R\'enyi 35}] {
		\includegraphics[width=0.4\textwidth]{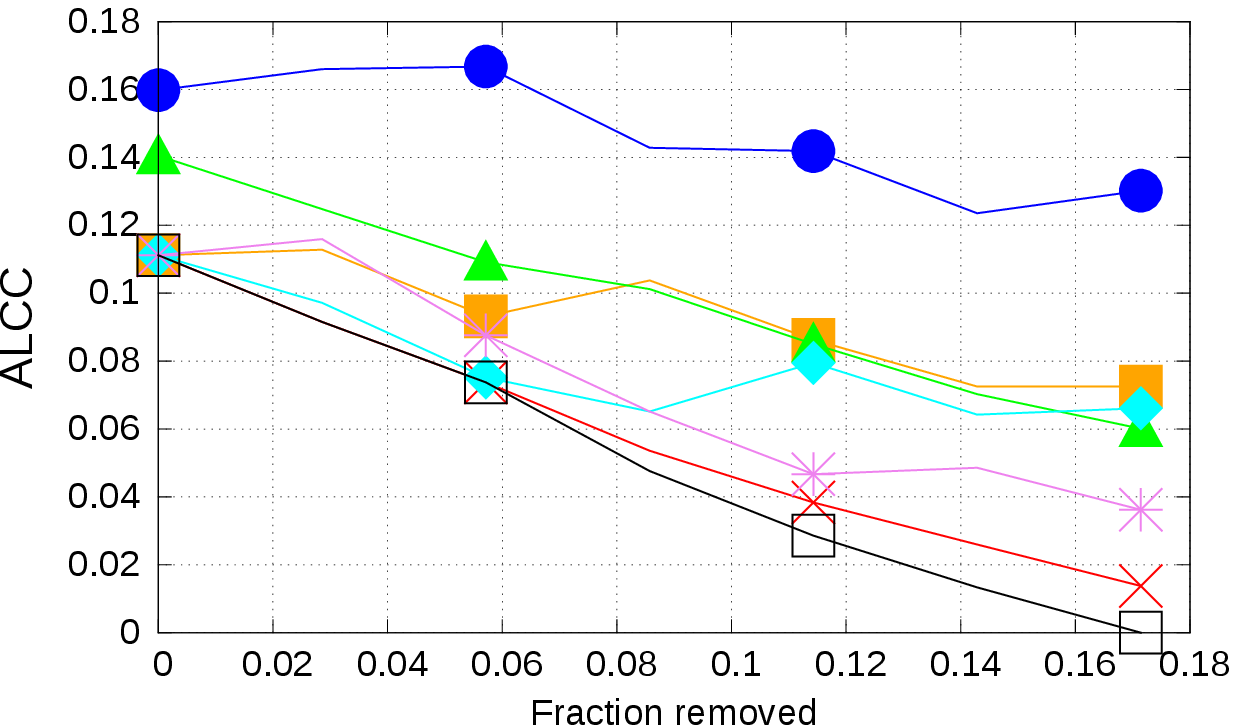}
		\label{fig:gccer35}
	}
  \subfigure[Method legends] {
		\includegraphics[scale=0.7]{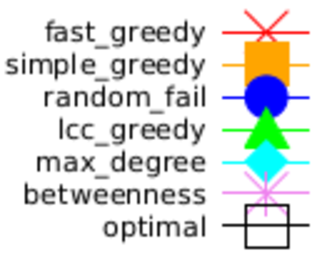}
		\label{fig:legend}}

	\caption{Average clustering coefficients (lower is better).}
	\label{fig:gcc}
	%\vspace{-0.25in}
	\end{figure*}
\begin{figure}[ht]
  \centering
  \subfigure[Critical nodes] {
    \includegraphics[scale=0.25]{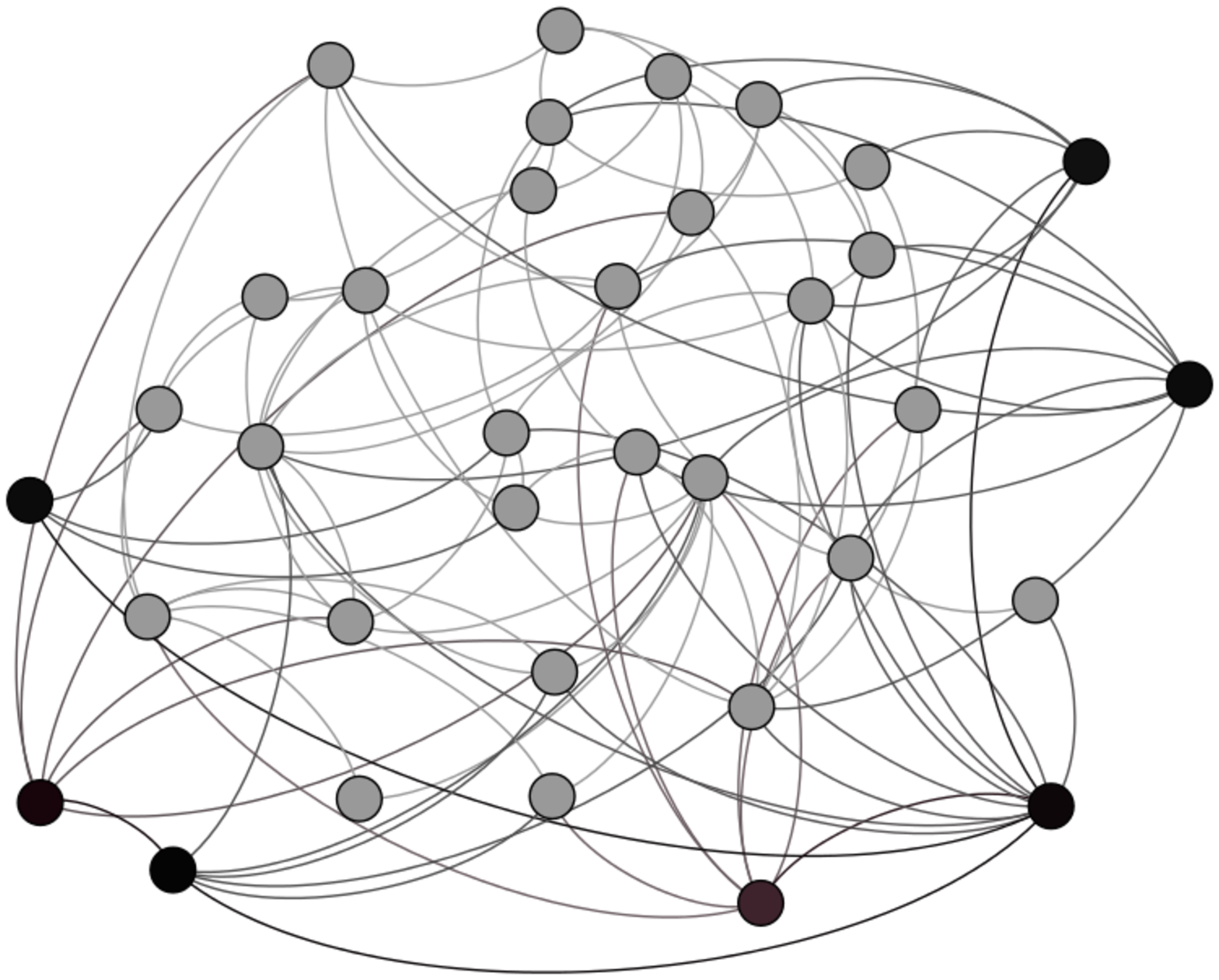}
  }
  \subfigure[Residual graph] {
    \includegraphics[scale=0.25]{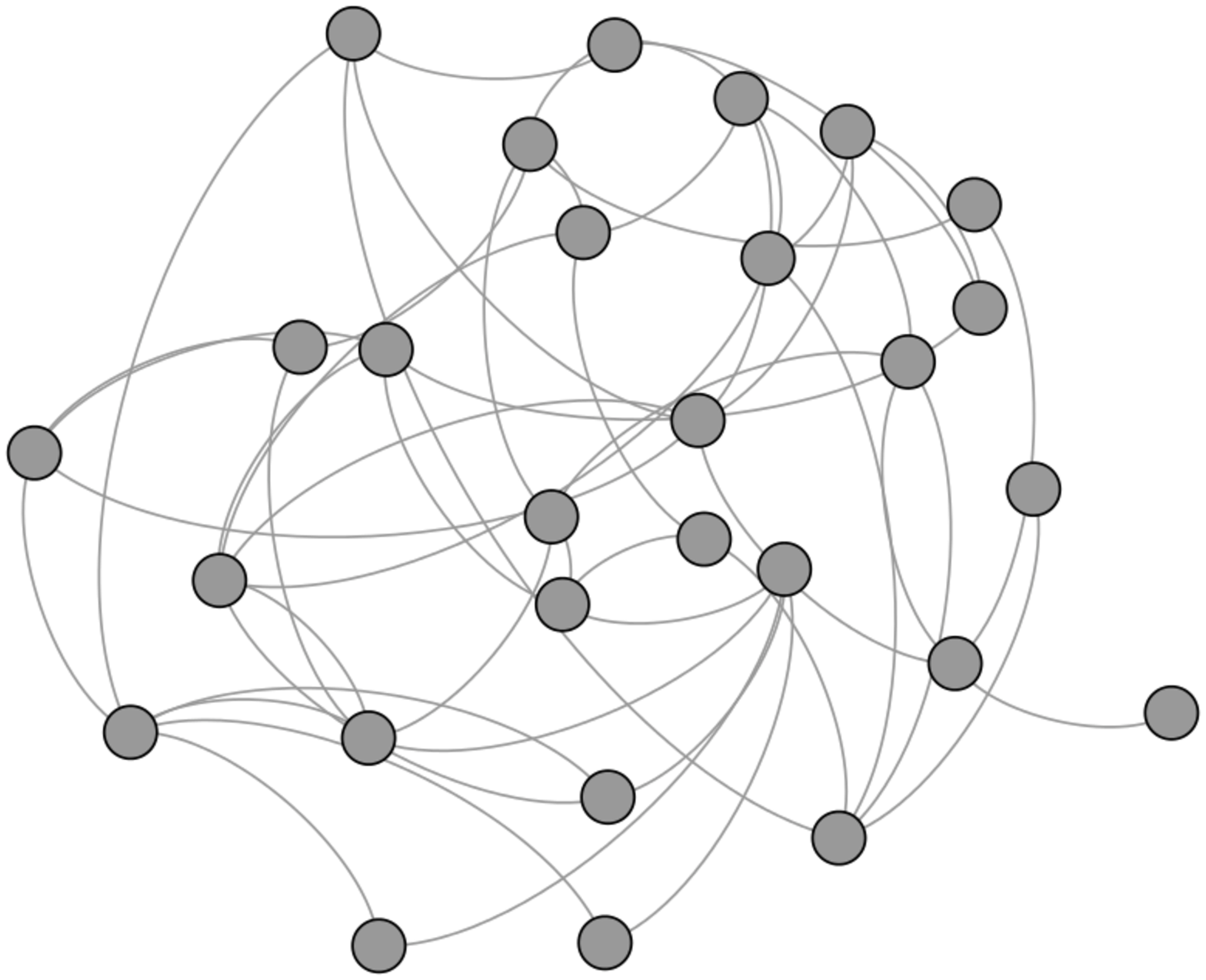}
  } \caption{ {\color{black}The optimal solution (black nodes) on the Erd\H{o}s-R\'enyi network with 35 nodes,
  with $k = 7$. Notice that the residual graph after removal of the optimal solution is triangle-free.}}
  \label{fig:er35opt}
\end{figure}
In this section, we present results on the efficacy of the various algorithms to
lower the ALCC.
We observe 
(1) the performance of our algorithms in view of other strategies, and more importantly 
(2) the critical behavior of clustering coefficient when 
crucial nodes are removed by different criteria. The empirical results on synthesized and real data are presented in Fig. \ref{fig:gcc}.

%--------------------------%
\begin{figure}[t]
  \centering
  \subfigure[Erd\H{o}s-R\'enyi] {
    \includegraphics[width=0.4\textwidth]{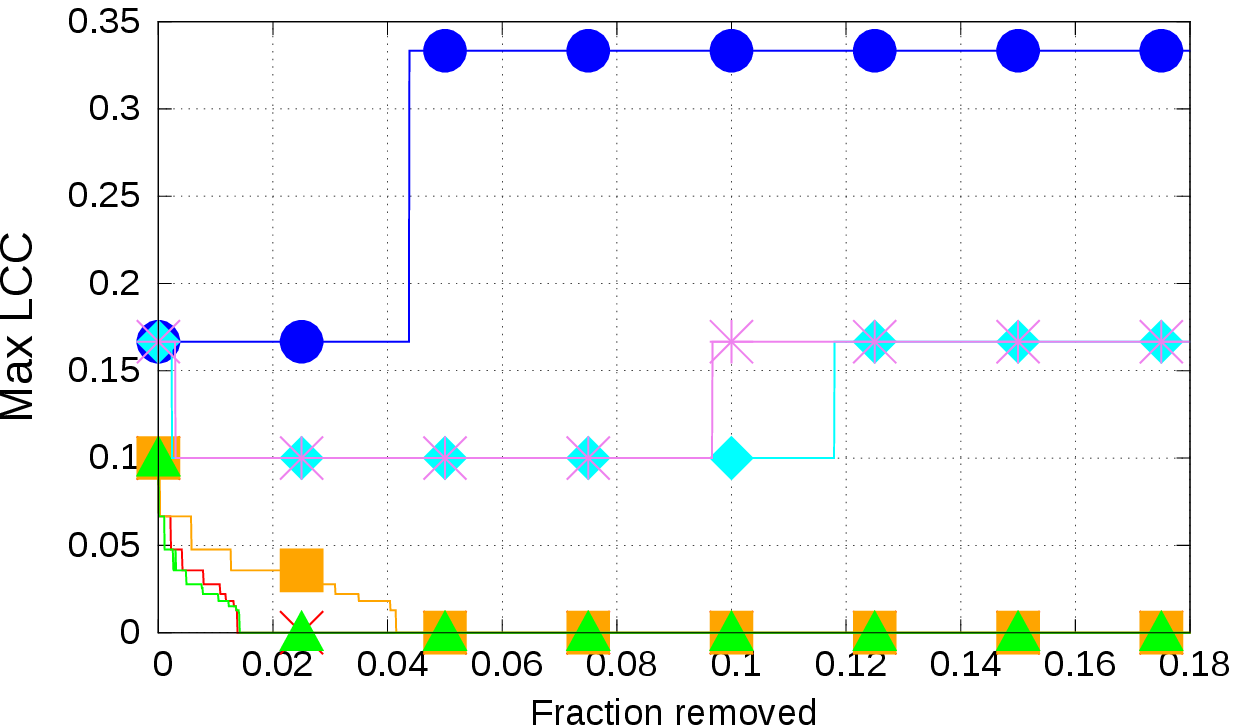}
    \label{fig:lccER} }
  \subfigure[Watz-Strogatz] {
    \includegraphics[width=0.4\textwidth]{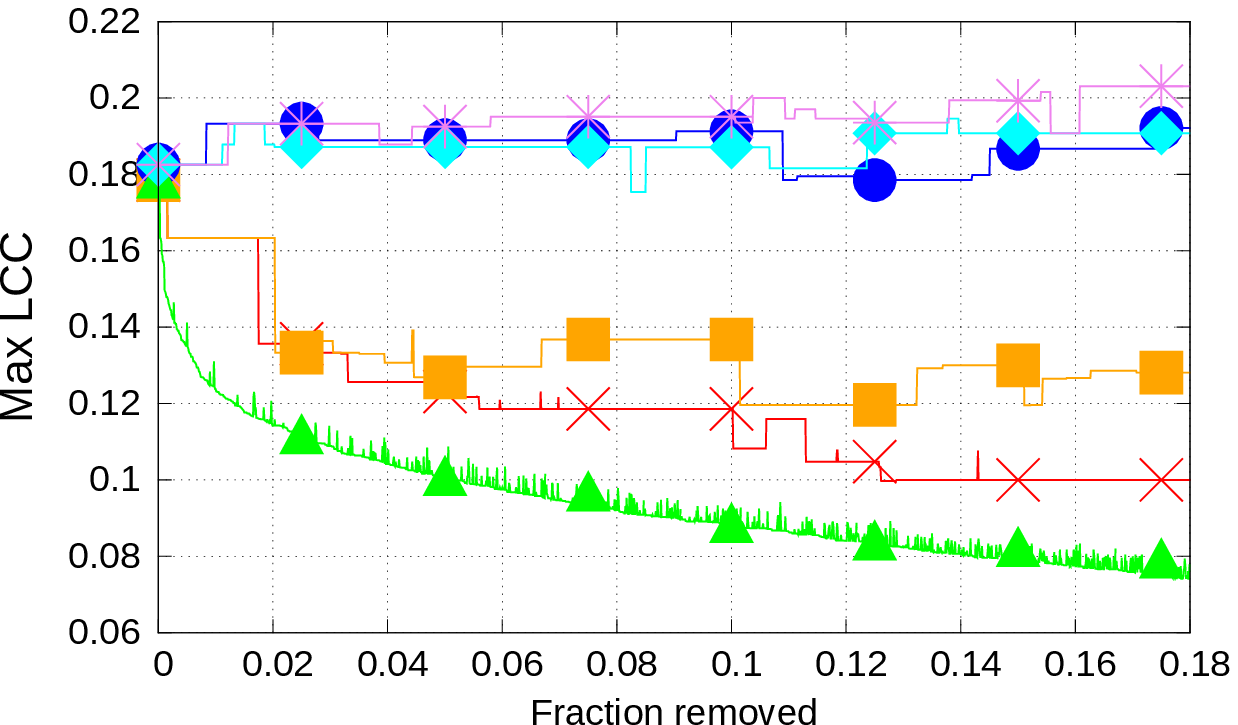} \label{fig:lccWS} }

  \subfigure[Barab\'asi-Albert] {
    \includegraphics[width=0.4\textwidth]{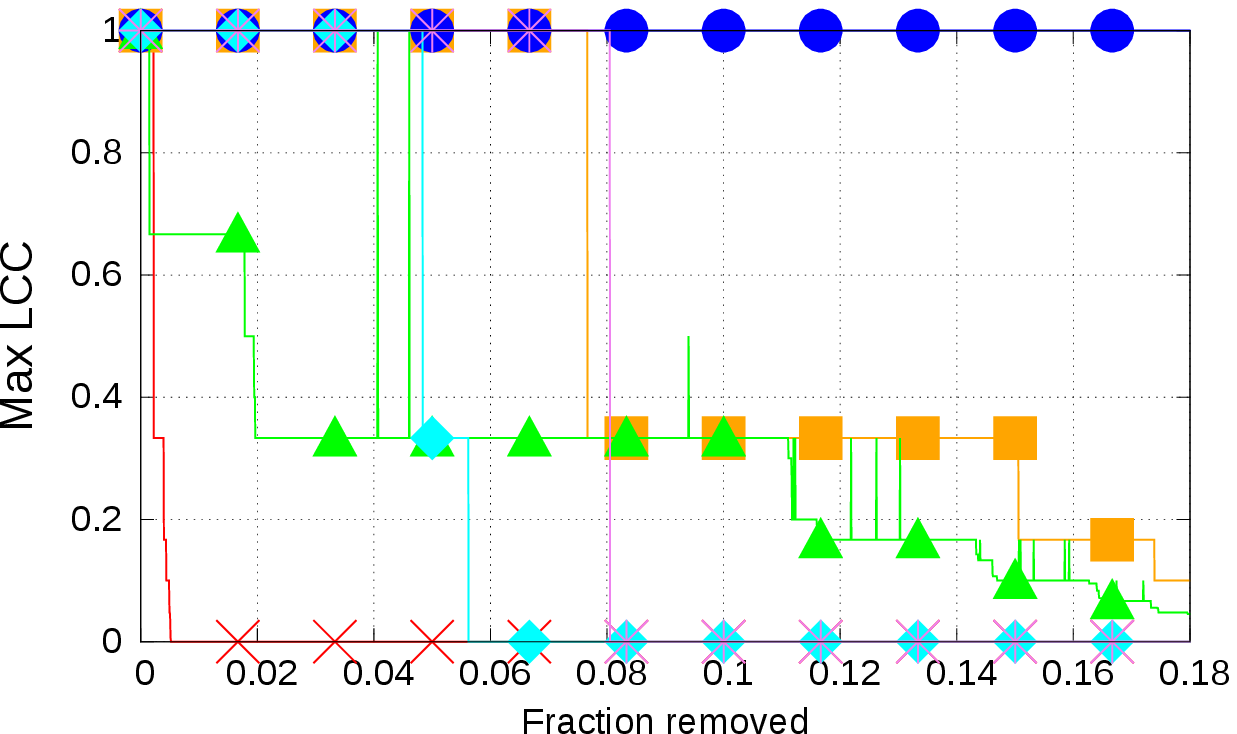} \label{fig:lccBA} }
  \subfigure[Arxiv] {
    \includegraphics[width=0.4\textwidth]{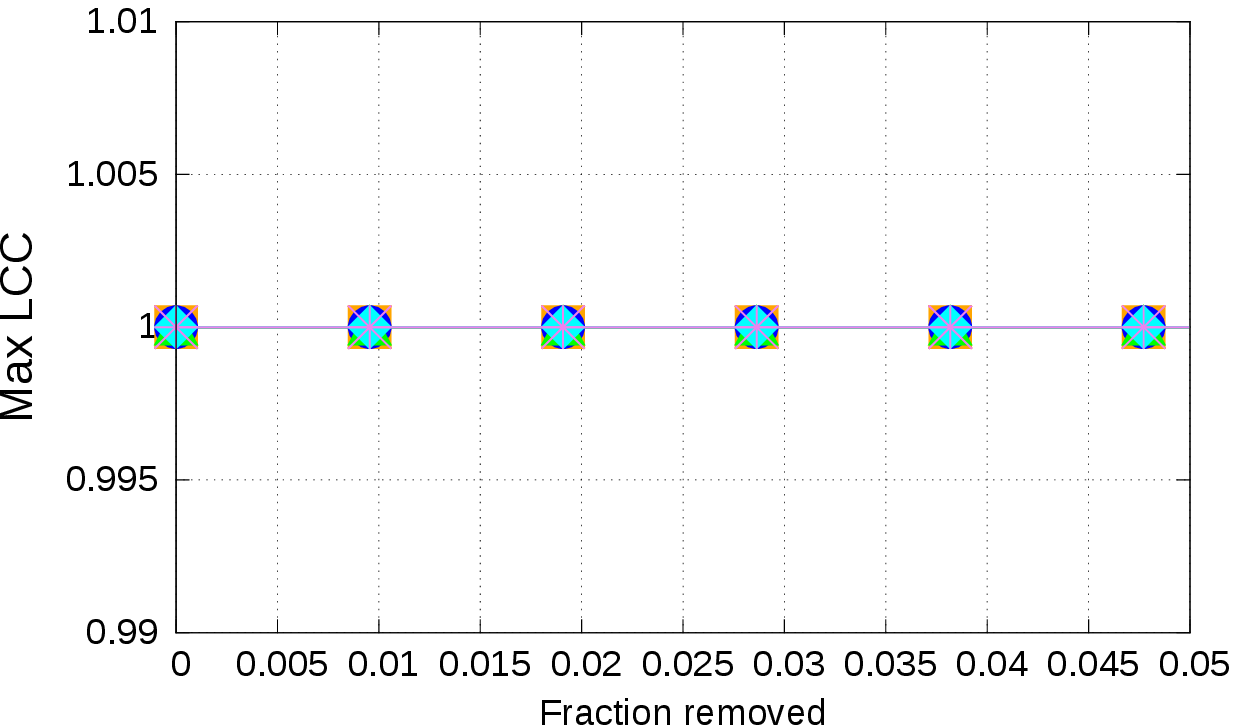} \label{fig:lccArxiv} }  

  \subfigure[NetHEPT] {
    \includegraphics[width=0.4\textwidth]{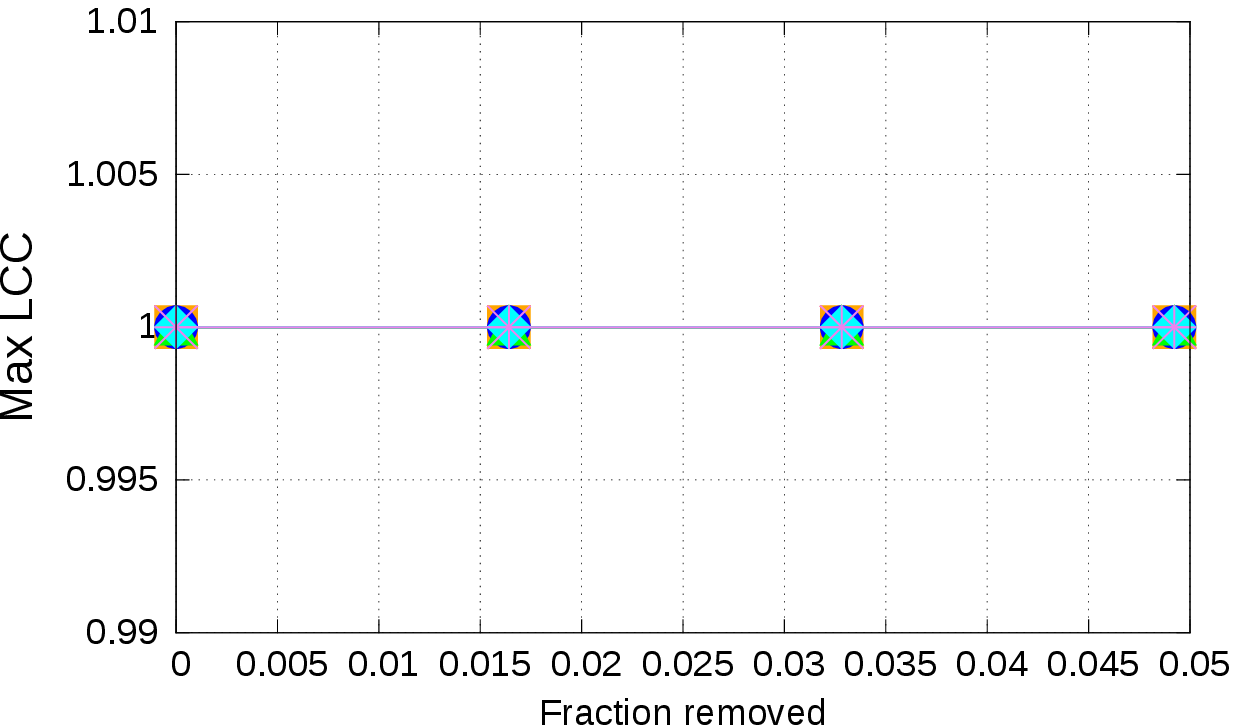} \label{fig:lccNET} }
  \subfigure[Facebook] {
    \includegraphics[width=0.4\textwidth]{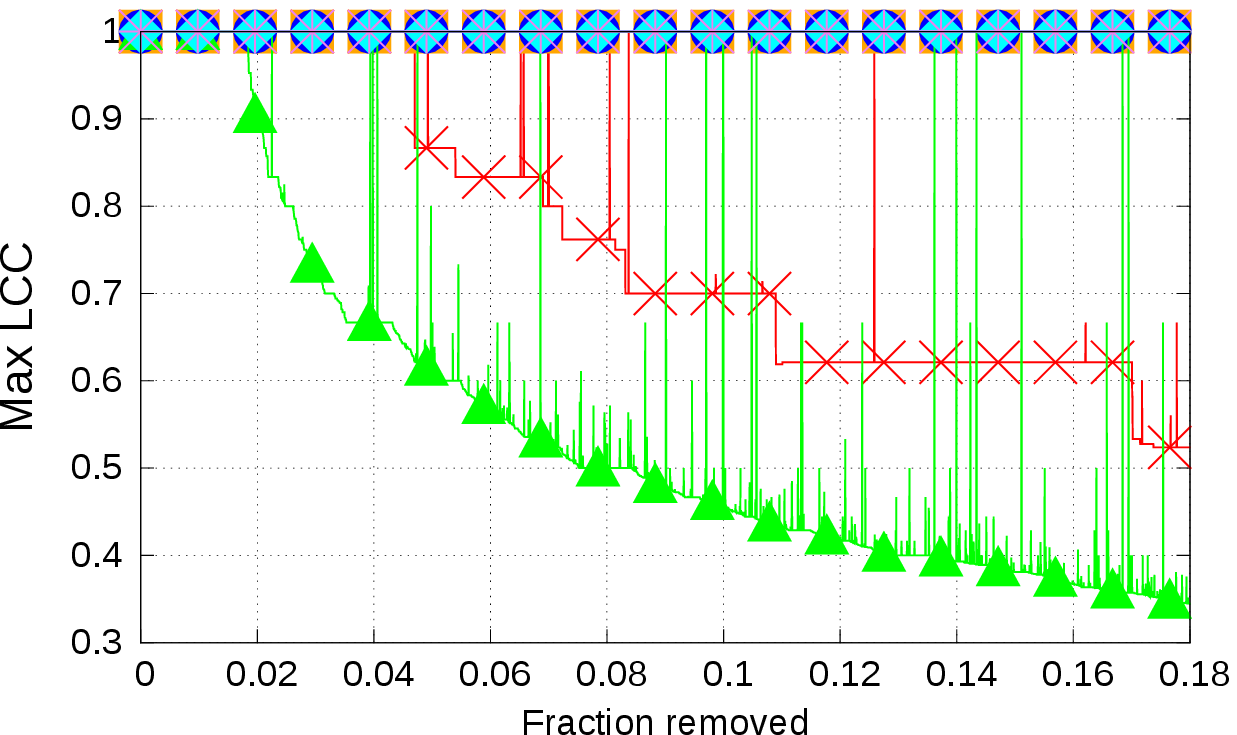}
    \label{fig:lccFB} }
  %\vspace{-0.1in}
  \caption{Maximum local clustering coefficients (lower is better). For legend, see Fig. \ref{fig:legend}.}
  \label{fig:lcc}
\end{figure}
%--------------------------%

As depicted in the subfigures, ALCC values produced by our algorithm \texttt{fast\_greedy} are consistently the best (lowest) values in all test cases, except in
the ER network with 35 nodes where \texttt{optimal}
was able to run. A visualization of the optimal solution
on this network for $k = 7$ is shown in Fig. \ref{fig:er35opt}. 
In the ER network with 10000 nodes, 
\texttt{fast\_greedy}, \texttt{lcc\_greedy} and \texttt{simple\_greedy} methods quickly destroy clustering as soon as 0.02 fraction of nodes (on \texttt{fast\_greedy} and \texttt{lcc\_greedy}) and 0.05 fraction of nodes (on \texttt{simple\_greedy}) 
are excluded from the networks.
Interestingly, \texttt{max\_degree} and \texttt{betweenness} methods do not appear much better 
than the baseline \texttt{random\_failure} 
method especially for \texttt{betweenness}. 
A possible explanation for this is the 
independence and equal probability of wiring edges in ER model. Moreover, because ER model neither generates triadic closures nor forms hubs, the network structure might be easily broken when a few random but important nodes are removed.

In WS model, we observe the same degrading behavior of ALCC value produced by all methods with \texttt{fast\_greedy} outperforming \texttt{lcc\_greedy} and \texttt{simple\_greedy} methods. Also in this model, these three methods outperformed the rest by a large magnitude. In BA model, \texttt{fast\_greedy} still performs best,
closely followed by \texttt{max\_degree} and \texttt{betweenness} methods. As BA model generates graphs with references given to the power-law distribution (i.e., forming hubs) the performance of \texttt{max\_degree} and \texttt{betweenness} can be explained. \texttt{lcc\_greedy} does not do well in this type of network as it takes a considerable fraction of total nodes in order to degrade the average clustering coefficient.

In conclusion, \texttt{fast\_greedy} is the best approach that consistently discovers nodes that are most important to the network clustering. The experiments also suggest that \texttt{max\_degree} and \texttt{betweenness}, despite their popularity, might not be ideal methods to analyze structural vulnerability of complex networks. In addition, these experiments also show that (1) ALCC isn't very susceptible to random failures, and (2) network clusters generated by ER, WS and BA can potentially be vulnerable to targeted attacks as the respective ALCC can quickly be impaired when only a few vertices are removed from the graphs.

In real data, the superior nature of \texttt{fast\_greedy} becomes more visible as it beats other strategies by a significant gap. In real traces, \texttt{max\_degree} and \texttt{betweenness} perform similarly while \texttt{lcc\_greedy} and \texttt{simple\_greedy} methods fluctuate in between. \texttt{random\_failure}, unsurprisingly, remains the worst. We observe that even in big real networks, \texttt{fast\_greedy} performs very well by degrading the ALCC dramatically (nearly 90\%, 33\% and 55\% of ALCC decrement on ArXiv, NetHEPT and Facebook) as more nodes are excluded from the data. This fact implies that those practical systems, despite their complex structure and functionality, commonly expose their clustering vulnerability to targeted or adversarial attacks. Our proposed approach \texttt{fast\_greedy} effectively discovers the critical nodes with high impact to those network structures. The results also demonstrate that \texttt{simple\_greedy} and \texttt{lcc\_greedy} are also good options though they require long execution time as we show 
below.
%--------------------------%
\subsection{Maximum Local Clustering Coefficient} \label{sect:exp_lcc}
We next examine the maximum local clustering coefficient (max-LCC) of nodes remaining 
in the residual graphs. This local measure is meaningful in the sense that a small max-LCC of a network indicates a low level of clustering. Therefore, we 
observe how the methods reduce the max LCC of the 
graphs. 
The results are reported in Fig. \ref{fig:lcc}. 
The subfigures indicate that \texttt{fast\_greedy} is really effective in not only degrading ALCC but also the max-LCC of all tested networks. In ER and BA models, \texttt{fast\_greedy} quickly destroys the clustering coefficients at just 0.02\% total nodes removed, and only lags behind \texttt{lcc\_greedy} (which was expected to be the leading method) in WS model and Facebook. Furthermore, \texttt{fast\_greedy} appears to be more stable than the others as it does not fluctuate between high and low values. In Facebook data, 
\texttt{fast\_greedy} quickly degrades max-LCC values from 1 to approximately 0.5. This fact indicates that the resulting Facebook clusters and structure might not be very robust. In ArXiv and NetHEPT data, all methods are unable to degrade the LCC which demonstrates that there are a lot of local clusters in these networks.
%--------------------------%
\subsection{Running Time} \label{sect:exp_rt}
%--------------------------%
\begin{figure}[t]
  \centering
	\subfigure[Edors-Renyi] {
		\includegraphics[scale=0.19]{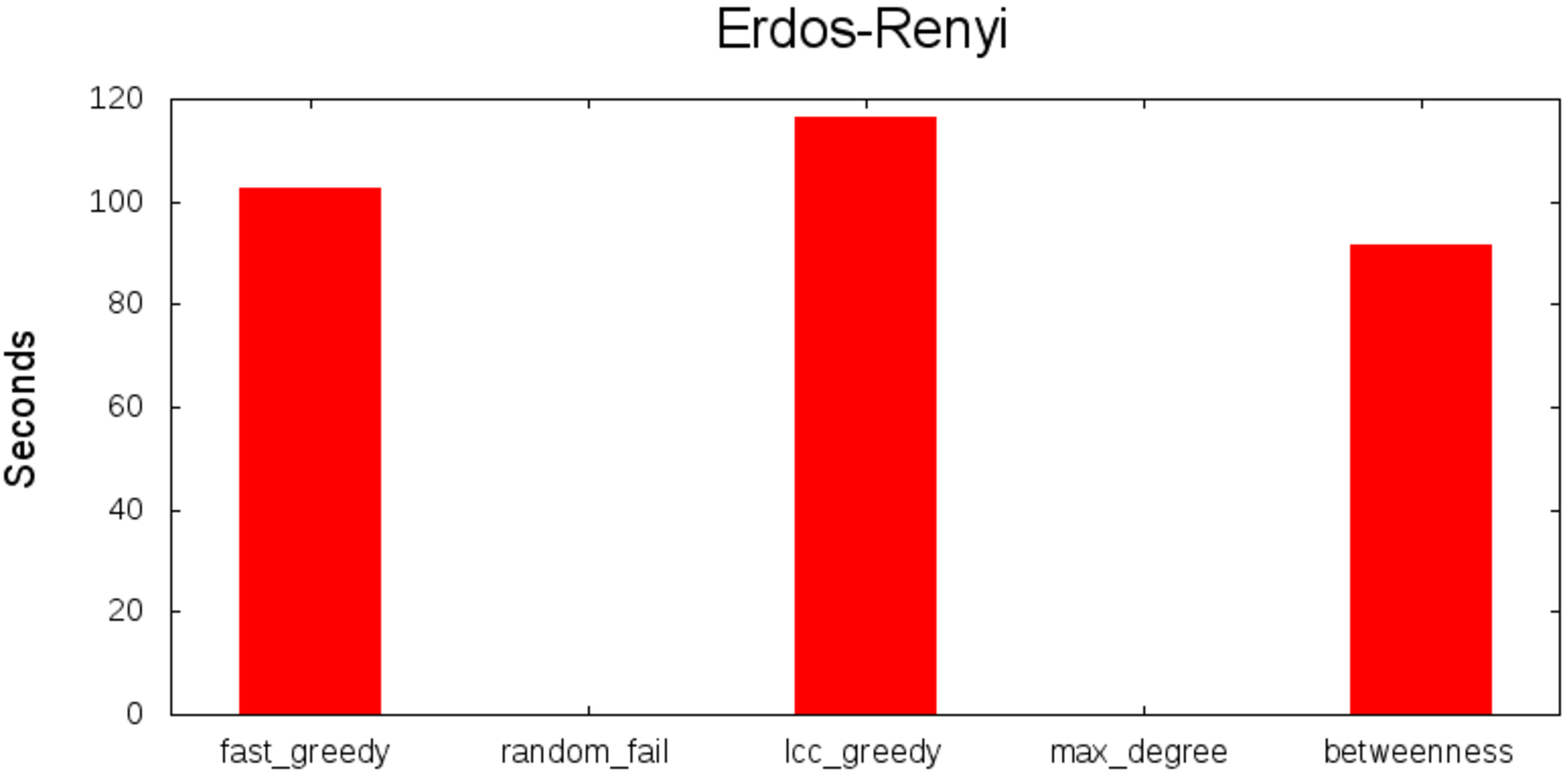}
		\label{fig:runtimeER}	}  
  \subfigure[Watz-Strogatz] {
		\includegraphics[scale=0.19]{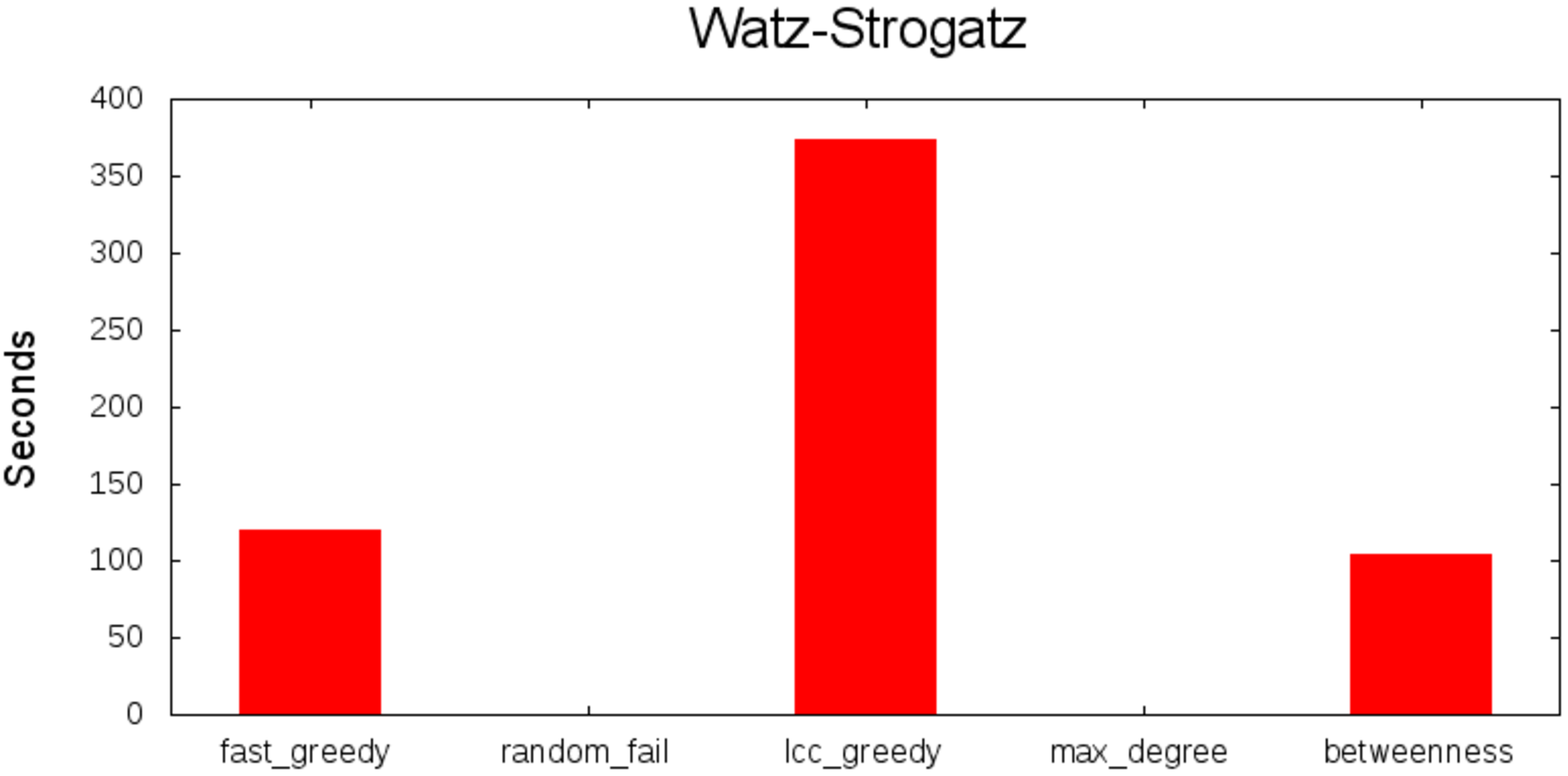}
		\label{fig:runtimeWS}	}
	\subfigure[Barab\'asi-Albert] {
		\includegraphics[scale=0.19]{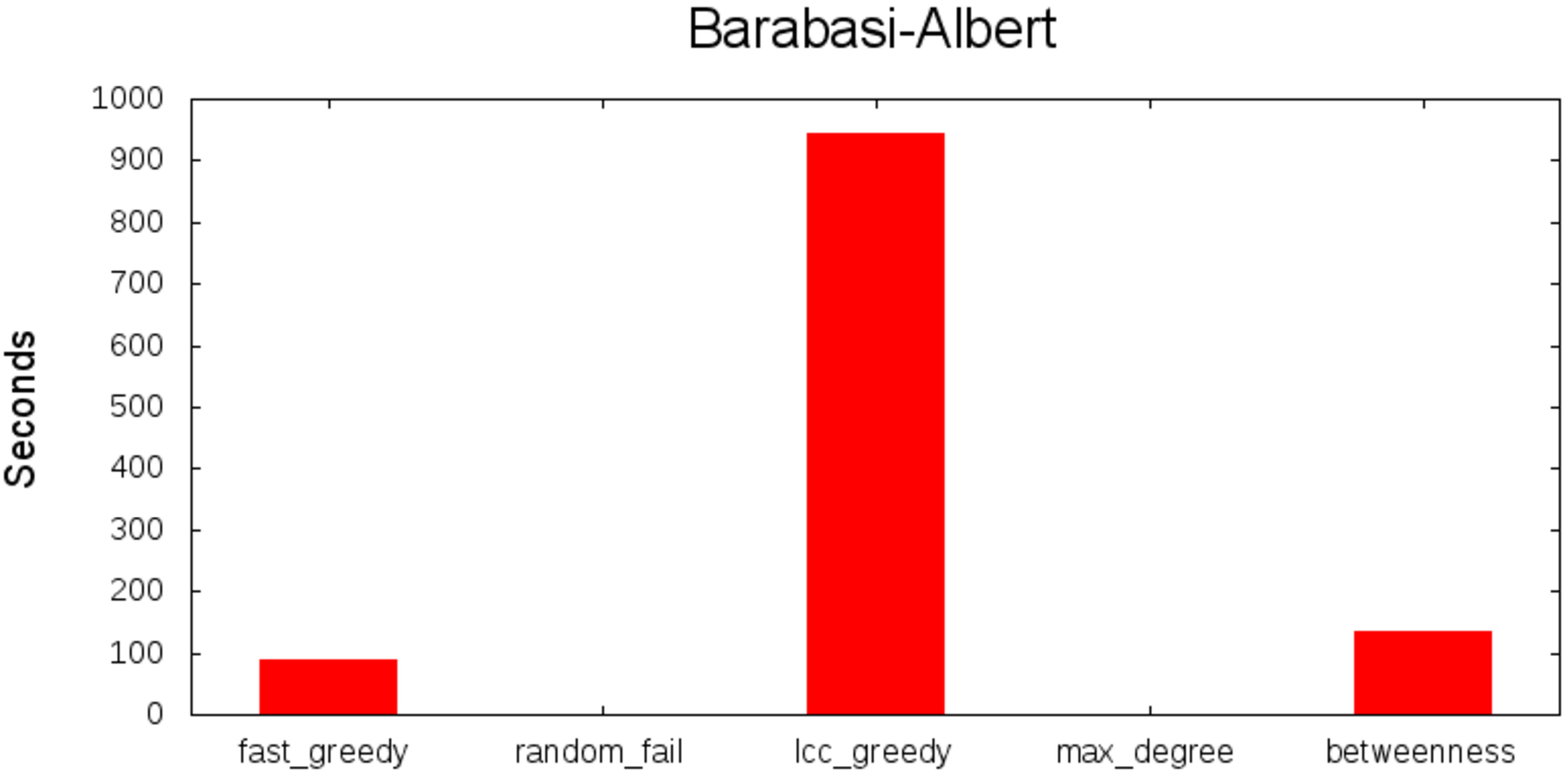}
		\label{fig:runtimeBA}	}
	\subfigure[Arxiv] {
		\includegraphics[scale=0.19]{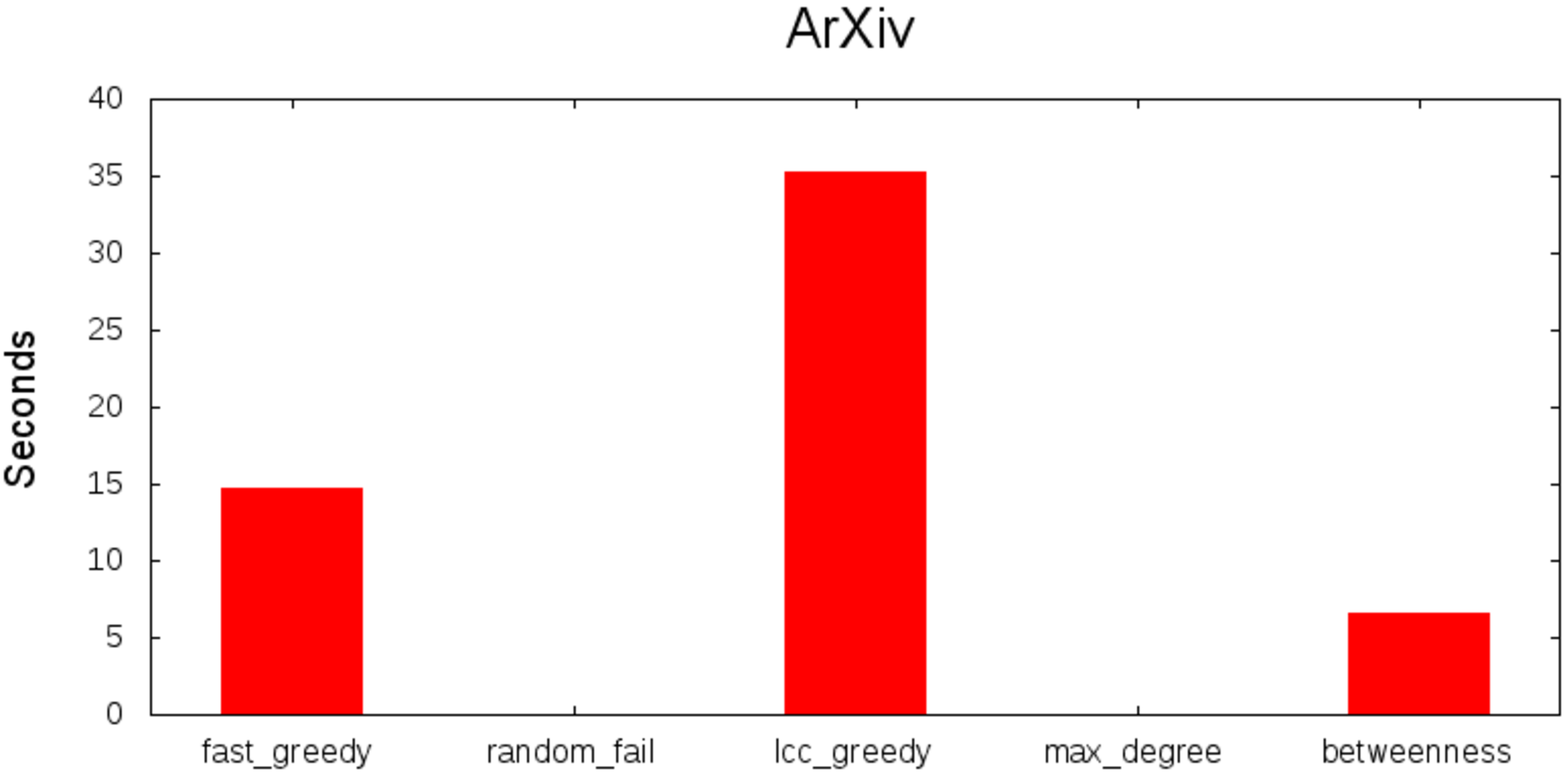}
		\label{fig:runtimeArxiv}	}  
  \subfigure[NetHEPT] {
		\includegraphics[scale=0.19]{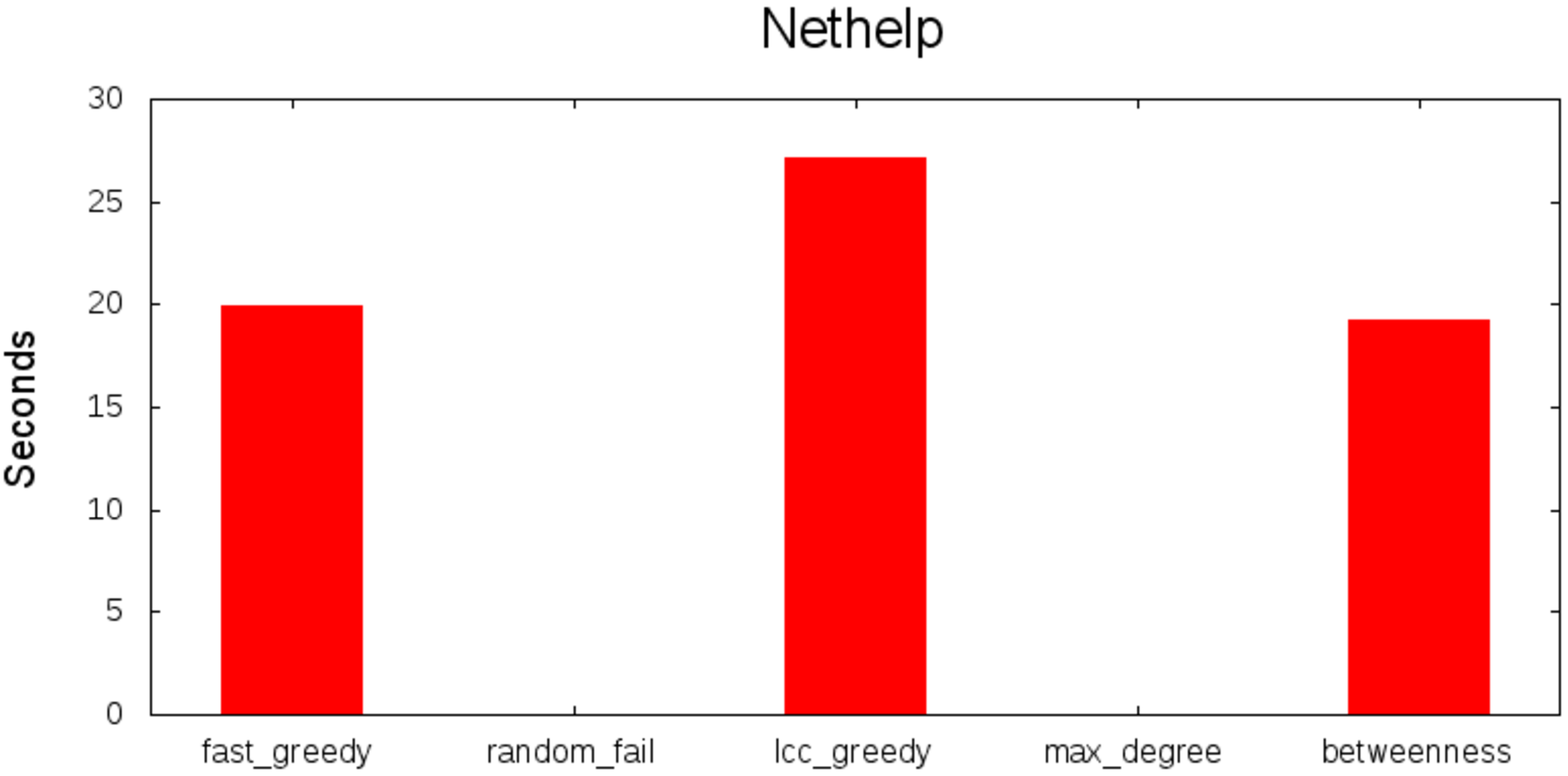}
		\label{fig:runtimeNET}	}
	\subfigure[Facebook] {
		\includegraphics[scale=0.19]{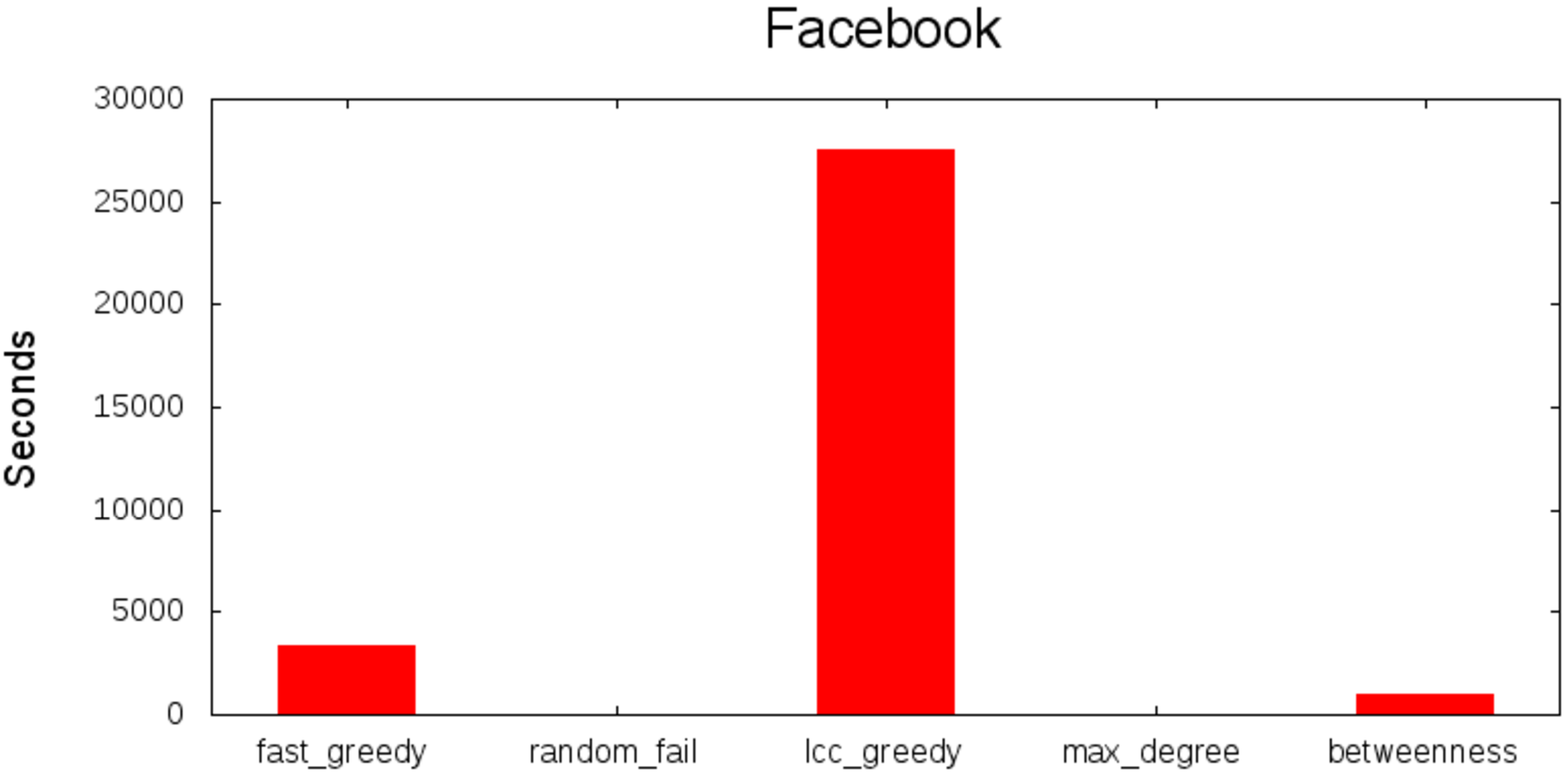}
		\label{fig:runtimeFB}	}
	%\vspace{-0.1in}	
	\caption{Running time}
	\label{fig:runtime}
	%\vspace{-0.3in}
	\end{figure}
%--------------------------%
The running time of all methods is presented in Fig. \ref{fig:runtime}. As the baseline methods, 
\texttt{random\_failure} and 
\texttt{max\_degree} do not require much time 
for their execution due to their simple 
nature whereas \texttt{lcc\_greedy},
in contrast, requires 
a considerable amount of execution time.
\texttt{fast\_greedy} and \texttt{betweenness} algorithms on average require fairly similar amounts of time for their tasks on all networks.
\texttt{simple\_greedy}, as a pay off for its simple design and implementation, takes a significant amount of time to finish its tasks (at least 5 times more than that taken by \texttt{lcc\_greedy}) and is excluded from the charts for more visibility.

\section{Conclusion}\label{ssconclusion}

Clustering vulnerability is an important aspect in assessing the robustness of complex networks,
as the level of clustering has significance 
for a variety of applications,
including a salient role in the propagation of information in a social network.
We have shown the discovery of the most 
important nodes 
to clustering is $NP$-complete, { \color{black}
and we offer two polynomial-time heuristics
for this identification. 
Empirical results in comparison with different strategies on synthesized and real networks show that the average clustering coefficient is robust to failure of random nodes 
and confirm that our suggested 
algorithm FAGA (\texttt{fast\_greedy}) 
is effective in analyzing node vulnerability 
of clustering and is scalable 
to larger networks. }

%\bibpunct{(}{)}{,}{a}{}{;}
\bibliographystyle{plainnat}
\bibliography{myref,triangle,mendeley}   % name your BibTeX data base

\begin{thebibliography}{39}
\providecommand{\natexlab}[1]{#1}
\providecommand{\url}[1]{\texttt{#1}}
\expandafter\ifx\csname urlstyle\endcsname\relax
  \providecommand{\doi}[1]{doi: #1}\else
  \providecommand{\doi}{doi: \begingroup \urlstyle{rm}\Url}\fi

\bibitem[Albert and Barab\'asi(2002)]{Albert2002}
R\'eka Albert and Albert-L\'aszl\'o Barab\'asi.
\newblock Statistical mechanics of complex networks.
\newblock \emph{Rev. Mod. Phys.}, 74:\penalty0 47--97, Jan 2002.
\newblock \doi{10.1103/RevModPhys.74.47}.

\bibitem[Albert et~al.(2000)Albert, Jeong, and Barab\'asi]{Albert00theinternet}
R\'eka Albert, Hawoong Jeong, and Albert-L\'aszl\'o Barab\'asi.
\newblock Error and attack tolerance of complex networks.
\newblock \emph{Nature}, 406:\penalty0 200--0, 2000.

\bibitem[Alim et~al.(2014{\natexlab{a}})Alim, Kunhle, and Thai]{alimASONAM14}
Md~Abdul Alim, Alan Kunhle, and My~T. Thai.
\newblock Are communities as strong as we think?
\newblock In \emph{Proceedings of the 2014 IEEE/WIC/ACM International
  Conference on Advances in Social Networks Analysis and Mining}, ASONAM '14,
  New York, NY, USA, 2014{\natexlab{a}}. ACM.

\bibitem[Alim et~al.(2014{\natexlab{b}})Alim, Nguyen, Thang, and Thai]{namWI14}
Md~Abdul Alim, Nam~P. Nguyen, Dinh~N. Thang, and My~T. Thai.
\newblock Structural vulnerability analysis of overlapping communities in
  complex networks.
\newblock In \emph{Proceedings of the 2014 IEEE/WIC/ACM International
  Conference on Web Intelligence}, WI '14, pages 231--235, New York, NY, USA,
  2014{\natexlab{b}}. ACM.

\bibitem[Allesina and Pascual(2009)]{allesina2009}
Stefano Allesina and Mercedes Pascual.
\newblock Googling food webs: Can an eigenvector measure species' importance
  for coextinctions?
\newblock \emph{PLoS Comput Biol}, 5\penalty0 (9):\penalty0 e1000494, 09 2009.
\newblock \doi{10.1371/journal.pcbi.1000494}.

\bibitem[Barclay et~al.(2013)Barclay, Edling, and Rydgren]{barclay2013peer}
Kieron~J Barclay, Christofer Edling, and Jens Rydgren.
\newblock Peer clustering of exercise and eating behaviours among young adults
  in sweden: a cross-sectional study of egocentric network data.
\newblock \emph{BMC public health}, 13\penalty0 (1):\penalty0 784, 2013.

\bibitem[Callaway et~al.(2000)Callaway, Newman, Strogatz, and
  Watts]{Callaway2000}
Duncan~S. Callaway, M.~E.~J. Newman, Steven~H. Strogatz, and Duncan~J. Watts.
\newblock Network robustness and fragility: Percolation on random graphs.
\newblock \emph{Phys. Rev. Lett.}, 85:\penalty0 5468--5471, Dec 2000.
\newblock \doi{10.1103/PhysRevLett.85.5468}.

\bibitem[Centola(2010)]{centola2010spread}
Damon Centola.
\newblock The spread of behavior in an online social network experiment.
\newblock \emph{Science}, 329\penalty0 (5996):\penalty0 1194--1197, 2010.

\bibitem[Centola(2011)]{centola2011experimental}
Damon Centola.
\newblock An experimental study of homophily in the adoption of health
  behavior.
\newblock \emph{Science}, 334\penalty0 (6060):\penalty0 1269--1272, 2011.

\bibitem[Chan et~al.(2014)Chan, Tong, and Akoglu]{hau2014}
Hau Chan, Hanghang Tong, and Leman Akoglu.
\newblock \emph{Make It or Break It: Manipulating Robustness in Large
  Networks}, chapter~37, pages 325--333.
\newblock SIAM, 2014.
\newblock \doi{10.1137/1.9781611973440.37}.

\bibitem[Chen et~al.(2010)Chen, Wang, and Wang]{chen10kdd}
W.~Chen, C.~Wang, and Y.~Wang.
\newblock Scalable influence maximization for prevalent viral marketing in
  large-scale social networks.
\newblock In \emph{Proceedings of the 16th ACM SIGKDD international conference
  on Knowledge discovery and data mining}, KDD, 2010.

\bibitem[Chen(2016)]{Chen2016}
Xin Chen.
\newblock {System vulnerability assessment and critical nodes identification}.
\newblock \emph{Expert Systems with Applications}, 65:\penalty0 212--220, 2016.
\newblock ISSN 09574174.
\newblock \doi{10.1016/j.eswa.2016.08.051}.
\newblock URL \url{http://dx.doi.org/10.1016/j.eswa.2016.08.051}.

\bibitem[Criado and Romance(2012)]{criado2012Strutural}
Regino Criado and Miguel Romance.
\newblock Structural vulnerability and robustness in complex networks:
  Different approaches and relationships between them.
\newblock In My~T. Thai and Panos~M. Pardalos, editors, \emph{Handbook of
  Optimization in Complex Networks}, Springer Optimization and Its
  Applications, pages 3--36. Springer New York, 2012.
\newblock ISBN 978-1-4614-0856-7.
\newblock \doi{10.1007/978-1-4614-0857-4_1}.

\bibitem[dataset(2003)]{arxivdataset}
ArXiv dataset.
\newblock http://www.cs.cornell.edu/projects/kddcup/datasets.html.
\newblock \emph{KDD Cup 2003}, Feb 2003.

\bibitem[Dinh et~al.(2012{\natexlab{a}})Dinh, Nguyen, and Thai]{Dinh2012c}
Thang~N. Dinh, D.~T. Nguyen, and My~T. Thai.
\newblock {Cheap, easy, and massively effective viral marketing in social
  networks: truth or fiction?}
\newblock In \emph{23rd ACM Conference on Hypertext and Social Media},
  2012{\natexlab{a}}.

\bibitem[Dinh et~al.(2012{\natexlab{b}})Dinh, Xuan, Thai, Pardalos, and
  Znati]{thangton2012}
Thang~N. Dinh, Ying Xuan, My~T. Thai, Panos~M. Pardalos, and Taieb Znati.
\newblock On new approaches of assessing network vulnerability: hardness and
  approximation.
\newblock \emph{IEEE/ACM Trans. Netw.}, 20\penalty0 (2):\penalty0 609--619,
  April 2012{\natexlab{b}}.
\newblock ISSN 1063-6692.
\newblock \doi{10.1109/TNET.2011.2170849}.

\bibitem[Dinh et~al.(2013)Dinh, Zhang, Nguyen, and Thai]{Dinh2013}
Thang~N. Dinh, Huiyuan Zhang, Dzung~T. Nguyen, and My~T. Thai.
\newblock {Cost-Effective Viral Marketing for Time-Critical Campaigns in
  Large-Scale Social Networks}.
\newblock \emph{Transactions on Networking}, 22\penalty0 (6):\penalty0
  2001--2011, 2013.

\bibitem[Erd\H{o}s and R\'enyi(1960)]{Erdos60onthe}
P.~Erd\H{o}s and A~R\'enyi.
\newblock On the evolution of random graphs.
\newblock In \emph{Publication of the Mathematical Institute of the Hungarian
  Academy of Sciences}, pages 17--61, 1960.

\bibitem[Ertem et~al.(2016)Ertem, Veremyev, and Butenko]{Ertem2016}
Zeynep Ertem, Alexander Veremyev, and Sergiy Butenko.
\newblock {Detecting large cohesive subgroups with high clustering coefficients
  in social networks}.
\newblock \emph{Social Networks}, 46:\penalty0 1--10, 2016.
\newblock ISSN 03788733.
\newblock \doi{10.1016/j.socnet.2016.01.001}.
\newblock URL \url{http://dx.doi.org/10.1016/j.socnet.2016.01.001}.

\bibitem[Fiedler(1973)]{Fiedler73}
M.~Fiedler.
\newblock {Algebraic connectivity of graphs}.
\newblock \emph{Czechoslovak Mathematical Journal}, 23\penalty0 (98):\penalty0
  298--305, 1973.

\bibitem[Frank and Frisch(1970)]{Frank1970}
H.~Frank and IT. Frisch.
\newblock Analysis and design of survivable networks.
\newblock \emph{Communication Technology, IEEE Transactions on}, 18\penalty0
  (5):\penalty0 501--519, October 1970.
\newblock ISSN 0018-9332.
\newblock \doi{10.1109/TCOM.1970.1090419}.

\bibitem[Gall(2014)]{Gall14}
F.~L. Gall.
\newblock Powers of tensors and fast matrix multiplication.
\newblock In \emph{Proceedings of the 39th International Symposium on
  International Symposium on Symbolic and Algebraic Computation}, ISSAC '14,
  New York, NY, USA, 2014. ACM.

\bibitem[Gomes et~al.(2016)Gomes, Esposito, Hutchison, Kuipers, Rak, and
  Tornatore]{Gomes2016}
Teresa Gomes, Christian Esposito, David Hutchison, Fernando Kuipers, Jacek Rak,
  and Massimo Tornatore.
\newblock {A survey of strategies for communication networks to protect against
  large-scale natural disasters}.
\newblock 2011:\penalty0 11--22, 2016.

\bibitem[Grubesic et~al.(2008)Grubesic, Matisziw, Murray, and
  Snediker]{grubesic08}
T.~H. Grubesic, T.~C. Matisziw, A.~T. Murray, and D.~Snediker.
\newblock Comparative approaches for assessing network vulnerability.
\newblock \emph{Inter. Regional Sci. Review}, 31, 2008.

\bibitem[Holme et~al.(2002)Holme, Kim, Yoon, and Han]{holme2002attack}
Petter Holme, Beom~Jun Kim, Chang~No Yoon, and Seung~Kee Han.
\newblock Attack vulnerability of complex networks.
\newblock \emph{Phys. Rev. E}, 65:\penalty0 056109, May 2002.

\bibitem[Kempe et~al.(2003)Kempe, Kleinberg, and Tardos]{Kempe2003}
David Kempe, Jon Kleinberg, and {\'{E}}va Tardos.
\newblock {Maximizing the spread of influence through a social network}.
\newblock \emph{Proceedings of the ninth ACM SIGKDD international conference on
  Knowledge discovery and data mining - KDD '03}, page 137, 2003.
\newblock ISSN 1557-2862.
\newblock \doi{10.1145/956755.956769}.
\newblock URL \url{http://portal.acm.org/citation.cfm?doid=956750.956769}.

\bibitem[Kuhnle et~al.(2017)Kuhnle, Pan, Alim, and Thai]{Kuhnle2017}
Alan Kuhnle, Tianyi Pan, Md~Abdul Alim, and My~T. Thai.
\newblock {Scalable Bicriteria Algorithms for the Threshold Activation Problem
  in Online Social Networks}.
\newblock In \emph{IEEE International Conference on Computer Communications},
  2017.

\bibitem[L{\"u} et~al.(2011)L{\"u}, Chen, and Zhou]{lu2011small}
Linyuan L{\"u}, Duan-Bing Chen, and Tao Zhou.
\newblock The small world yields the most effective information spreading.
\newblock \emph{New Journal of Physics}, 13\penalty0 (12):\penalty0 123005,
  2011.

\bibitem[Malik and Mucha(2013)]{malik2013role}
Nishant Malik and Peter~J Mucha.
\newblock Role of social environment and social clustering in spread of
  opinions in coevolving networks.
\newblock \emph{Chaos: An Interdisciplinary Journal of Nonlinear Science},
  23\penalty0 (4):\penalty0 043123, 2013.

\bibitem[Nguyen et~al.(2010)Nguyen, Xuan, and Thai]{Nguyen2010}
N.~P. Nguyen, Y.~Xuan, and M.~T. Thai.
\newblock {A novel method for worm containment on dynamic social networks}.
\newblock In \emph{Military Communications Conference}, pages 2180--2185, 2010.

\bibitem[Nguyen et~al.(2011)Nguyen, Dinh, Tokala, and Thai]{Nguyen2011}
Nam~P. Nguyen, Thang~N. Dinh, Sindhura Tokala, and My~T. Thai.
\newblock Overlapping communities in dynamic networks: their detection and
  mobile applications.
\newblock In \emph{Proceedings of the 17th annual international conference on
  Mobile computing and networking}, MobiCom '11, pages 85--96, New York, NY,
  USA, 2011. ACM.
\newblock ISBN 978-1-4503-0492-4.
\newblock \doi{10.1145/2030613.2030624}.

\bibitem[Nguyen et~al.(2013)Nguyen, Alim, Shen, and Thai]{namASONAM13}
Nam~P. Nguyen, Md~Abdul Alim, Yilin Shen, and My~T. Thai.
\newblock Assessing network vulnerability in a community structure point of
  view.
\newblock In \emph{Proceedings of the 2013 IEEE/ACM International Conference on
  Advances in Social Networks Analysis and Mining}, ASONAM '13, pages 231--235,
  New York, NY, USA, 2013. ACM.
\newblock ISBN 978-1-4503-2240-9.
\newblock \doi{10.1145/2492517.2492644}.

\bibitem[Peixoto and Bornholdt(2012)]{peixoto2012}
Tiago~P. Peixoto and Stefan Bornholdt.
\newblock Evolution of robust network topologies: Emergence of central
  backbones.
\newblock \emph{CoRR}, abs/1205.2909, 2012.

\bibitem[Ponton et~al.(2013)Ponton, Wei, and Sun]{Ponton2013}
J.~Ponton, Peng Wei, and Dengfeng Sun.
\newblock Weighted clustering coefficient maximization for air transportation
  networks.
\newblock In \emph{Control Conference (ECC), 2013 European}, pages 866--871,
  July 2013.

\bibitem[Schank and Wagner(2005)]{Schank05}
T.~Schank and D.~Wagner.
\newblock Finding, counting and listing all triangles in large graphs, an
  experimental study.
\newblock In \emph{Proc. of the 4th Int. Conf. on Experimental and Efficient
  Algorithms}, WEA'05, pages 606--609, Berlin, Heidelberg, 2005.
  Springer-Verlag.
\newblock ISBN 3-540-25920-1, 978-3-540-25920-6.
\newblock \doi{10.1007/11427186_54}.

\bibitem[Veremyev et~al.(2014)Veremyev, Prokopyev, and Pasiliao]{Veremyev2014}
Alexander Veremyev, Oleg~A. Prokopyev, and Eduardo~L. Pasiliao.
\newblock {An integer programming framework for critical elements detection in
  graphs}.
\newblock \emph{Journal of Combinatorial Optimization}, 28\penalty0
  (1):\penalty0 233--273, 2014.
\newblock ISSN 15732886.
\newblock \doi{10.1007/s10878-014-9730-4}.

\bibitem[Veremyev et~al.(2015)Veremyev, Prokipyev, and Pasiliao]{Veremyev2015}
Alexander Veremyev, Oleg~A. Prokipyev, and Eduardo~L. Pasiliao.
\newblock {Critical Nodes for Distance-Based Connectivity and Related Problems
  in Graphs}.
\newblock \emph{Networks}, 2015.
\newblock ISSN 1097-0037.
\newblock \doi{10.1002/net}.

\bibitem[Viswanath et~al.(2009)Viswanath, Mislove, Cha, and
  Gummadi]{facebookdataset}
B.~Viswanath, A.~Mislove, M.~Cha, and K.~P. Gummadi.
\newblock On the evolution of user interaction in facebook.
\newblock In \emph{2nd ACM SIGCOMM Workshop on Social Networks}, 2009.

\bibitem[Watts and Strogatz(1998)]{watts1998cds}
D.~J. Watts and S.~H. Strogatz.
\newblock {Collective dynamics of 'small-world' networks.}
\newblock \emph{Nature}, 393\penalty0 (6684):\penalty0 409--10, 1998.

\end{thebibliography}
\end{document}